\documentclass[acmsmall,review=false,anonymous=false]{acmart}
\usepackage{cleveref}
\usepackage{subcaption}
\usepackage{mathpartir}
\usepackage{proof}
\usepackage{fdsymbol} 
\usepackage{tasks}
\usepackage{wrapfig}

\newcommand{\ssj}[2]{\ensuremath{#1 \mapsto #2}}
\newcommand{\msj}[2]{\ensuremath{#1 \stackrel{*}{\mapsto} #2}}
\newcommand{\bsj}[2]{\ensuremath{#1 \Downarrow #2}}
\newcommand{\bstj}[2]{\ensuremath{#1 \Downmapsto #2}}
\newcommand{\bstej}[3]{\ensuremath{#1 \Downmapsto #2 \rightsquigarrow #3}}
\newcommand{\bstjp}[2]{\ensuremath{#1 \downmapsto #2}}
\newcommand{\bstejp}[3]{\ensuremath{#1 \downmapsto #2 \rightsquigarrow #3}}
\newcommand{\nontinf}{\ensuremath{\infty}}
\newcommand{\bsntj}[1]{\ensuremath{\bsj {#1} {\nontinf}}}
\newcommand{\ppj}[2]{\ensuremath{#1 \downarrow #2}}
\newcommand{\ssej}[3]{\ensuremath{#1 \mapsto #2 \rightsquigarrow #3}}
\newcommand{\msej}[3]{\ensuremath{#1 \stackrel{*}{\mapsto} #2 \rightsquigarrow #3}}
\newcommand{\bsej}[3]{\ensuremath{#1 \Downarrow #2 \rightsquigarrow #3}}
\newcommand{\nont}{\ensuremath{\Box}}

\newcommand{\eff}[2]{\ensuremath{\mathtt{eff}\{#1; #2\}}}
\newcommand{\val}[1]{\ensuremath{#1\;\mathit{val}}}
\newcommand{\Lam}[3]{\ensuremath{\mathtt{fun}\{#1, #2.\,#3\}[\,]}}
\newcommand{\lam}[3]{\ensuremath{\mathtt{fun}\{#1, #2.\,#3\}}}
\newcommand{\app}[2]{\ensuremath{\mathtt{app}[#1, #2]}}
\newcommand{\suc}[1]{\ensuremath{\mathtt{S[}#1\mathtt{]}}}
\newcommand{\Z}{\ensuremath{\mathtt{Z}[\,]}}
\newcommand{\z}{\ensuremath{\mathtt{Z}}}
\newcommand{\ifz}[4]{\ensuremath{\mathtt{case}\{#1;#2.#3\}[#4]}}

\newcommand{\elet}[3]{\ensuremath{\mathtt{let}\,#1\,=\,#2\,\mathtt{in}\,#3}}

\newcommand{\ty}[3]{\ensuremath{#1 \vdash #2 : #3}}
\newcommand{\nat}{\ensuremath{\mathtt{nat}}}
\newcommand{\fun}[2]{\ensuremath{#1 \rightarrow #2}}

\newcommand{\stack}[1]{\ensuremath{#1\;\mathit{stack}}}
\newcommand{\kframe}[1]{\ensuremath{#1\;\mathit{frame}}}
\newcommand{\frameframe}[1]{\langle #1 \rangle}
\newcommand{\fsucc}{\ensuremath{\frameframe {\suc -}}}
\newcommand{\fcase}[3]{\ensuremath{\frameframe {\ifz {#1} {#2} {#3} -}}}
\newcommand{\ffun}[1]{\ensuremath{\frameframe {\app - {#1}}}}
\newcommand{\farg}[1]{\ensuremath{\frameframe {\app {#1} -}}}
\newcommand{\estate}[2]{\ensuremath{#1 \rhd #2}}
\newcommand{\vstate}[2]{\ensuremath{#1 \lhd #2}}
\newcommand{\ksj}[2]{\ensuremath{#1 \mathop{{\mapsto}_K} } #2}
\newcommand{\kmj}[2]{\ensuremath{#1  \stackrel{*}{\mapsto}_K}  #2}
\newcommand{\kesj}[3]{\ensuremath{#1 \mathop{{\mapsto}_K} #2  \rightsquigarrow #3}}
\newcommand{\kemj}[3]{\ensuremath{#1 \stackrel{*}{\mapsto}_K} #2  \rightsquigarrow #3}

\newtheorem{theorem}{Theorem}
\newtheorem{lemma}[theorem]{Lemma}
\newtheorem{corollary}[theorem]{Corollary}

\newcommand{\jan}[1]{{\color{blue}Jan says: #1}}
\newcommand{\david}[1]{{\color{purple}David says: #1}}

\newcommand{\iwhile}[2]{\ensuremath{\mathtt{while}\{#1\}\mathtt{do}\{#2\}}}
\newcommand{\iset}[2]{\ensuremath{#1 := #2}}
\newcommand{\iseq}[2]{\ensuremath{#1;#2}}
\newcommand{\iif}[2]{\ensuremath{\mathtt{if}\{#1\}\mathtt{then}\{#2\}}}

\newcommand{\iskip}{\ensuremath{\mathtt{skip}}}

\newcommand{\hole}{\ensuremath{\langle\,\rangle}}

\newsavebox{\logoagdabox}
\sbox{\logoagdabox}{
  \raisebox{-2pt}{\includegraphics[height=1em]{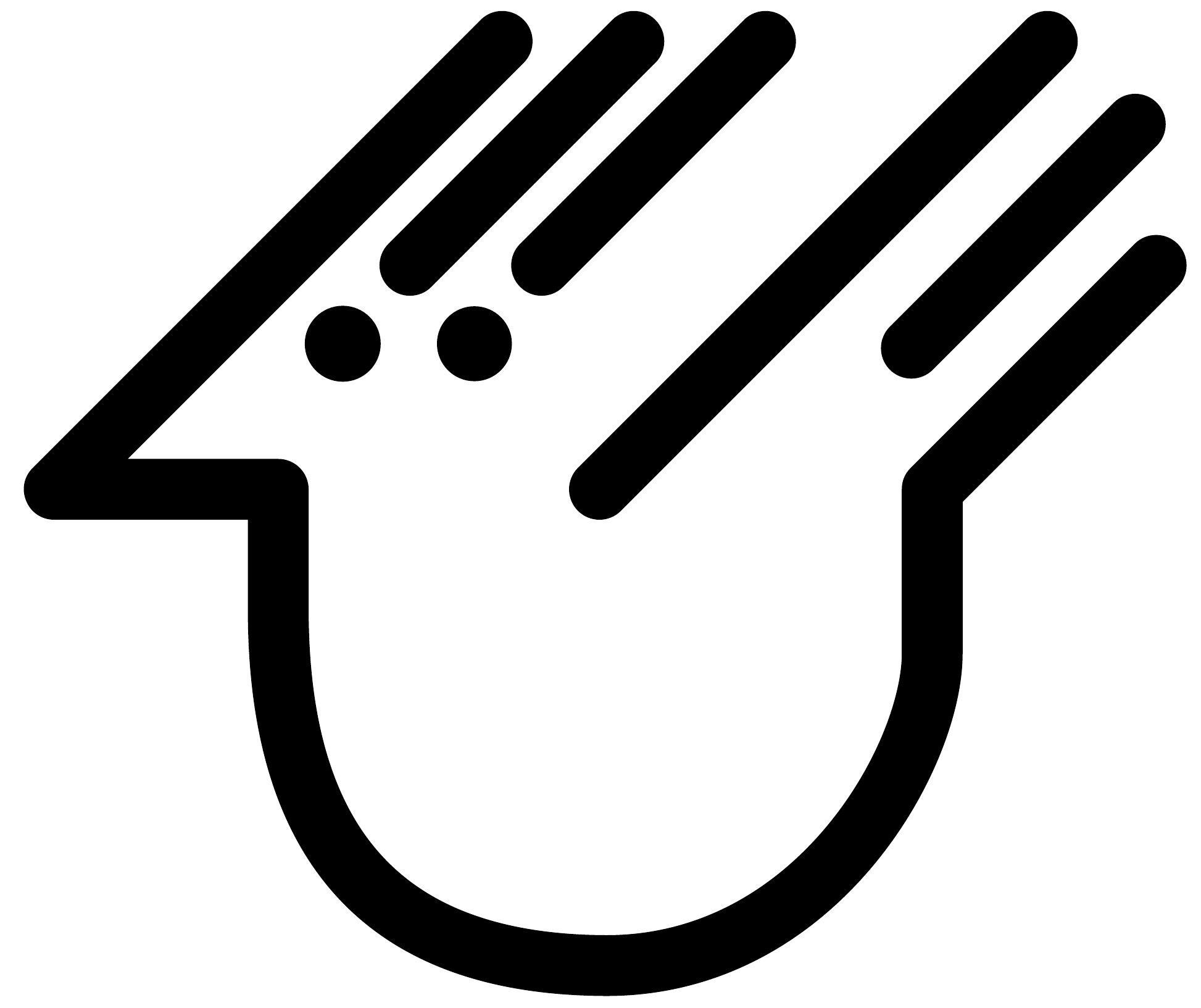}}
}
\newcommand{\agda}{\usebox{\logoagdabox}}

\AtBeginDocument{%
  }

\setcopyright{acmlicensed}
\copyrightyear{2018}
\acmYear{2018}
\acmDOI{XXXXXXX.XXXXXXX}

\begin{document}

\title{Big-Stop Semantics}
\subtitle{Small-Step Semantics 
 in a Big-Step Judgment}

\author{David M Kahn}
\email{kahnd@denison.edu}
\affiliation{%
  \institution{Denison University}
  \city{Granville}
  \state{Ohio}
  \country{USA}
}

\author{Jan Hoffmann}
\email{jhoffmann@cmu.edu}
\affiliation{%
  \institution{Carnegie Mellon University}
  \city{Pittsburgh}
  \state{Pennsylvania}
  \country{USA}
}

\author{Runming Li}
\email{runming@cmu.edu}
\affiliation{%
  \institution{Carnegie Mellon University}
  \city{Pittsburgh}
  \state{Pennsylvania}
  \country{USA}
}

\renewcommand{\shortauthors}{Kahn, Hoffmann, \& Li}

\begin{abstract}
  As is evident in the programming language literature, many
  practitioners favor specifying dynamic program behavior using
  big-step over small-step semantics.
  Unlike small-step semantics, which must dwell on every
  intermediate program state, big-step semantics conveniently jumps
  directly to the ever-important result of the computation.
  Big-step semantics also typically involves fewer inference rules than
  their small-step counterparts.
  However, in exchange for ergonomics, big-step semantics gives up
  power: Small-step semantics describes program behaviors that
  are outside the grasp of big-step semantics, notably divergence.
  
  This work presents a little-known extension of
  big-step semantics with inductive definitions that captures
  diverging computations without introducing error states. 
  This \emph{big-stop} semantics is illustrated for typed, untyped,
  and effectful variants of PCF, as well as a while-loop-based 
  imperative language.
  Big-stop semantics extends the standard big-step inference rules
  with a few additional rules to define an evaluation judgment
  that is equivalent to the reflexive-transitive 
  closure of small-step transitions.
  This simple extension contrasts 
  with other solutions in the literature that
  sacrifice ergonomics by
  introducing many additional inference rules, global state, and/or 
  less-commonly-understood reasoning principles like coinduction.
  The ergonomics of big-stop semantics is exemplified
  via concise Agda proofs for some key 
  results and compilation theorems.

\end{abstract}



\maketitle



\section{Introduction}\label{sec:intro}

Operational semantics is a popular choice for defining the dynamic
behavior of programs, both in the programming language literature and
for precisely defining industrial-strength languages such as
Standard ML~\cite{MilnerTH1990}, C~\cite{Norrish99,BlazyL09,EllisonR12}, or
Web Assembly~\cite{WebAssembly,Watt18}.
Advantages of operational semantics include that they are accessible
to non-experts without in-depth mathematical training and that they
scale well to describe novel and advanced language features.

Operational semantics can be characterized as belonging to
two main styles:
The small-step style defines a
step relation between the states of a transition
system in which the states correspond 
to programs.
It originates from Plotkin's structural operational semantics
(SOS)~\cite{plotkin1981structural} and encompasses abstract
machines~\cite{Gabbrielli2010}.
The outcome of an evaluation of a program is a sequence of transitions
that starts at the state that corresponds to the program.
Transitions are usually defined inductively using syntax-directed
inference rules.

The big-step style, first proposed by Kahn~\cite{kahn1987natural},
defines a relation between programs and outcomes.
Outcomes often consist of a returned value and potentially an effect,
for instance, the cost of the evaluation or a sequence of I/O
operations.
A big-step operational semantics does not define a notion of a step
and is thus also---perhaps more appropriately---called evaluation
dynamics \cite{harper2016practical} or natural semantics \cite{kahn1987natural}.
Like transitions in a small-step semantics, the big-step evaluation
relation is usually defined inductively with syntax-directed inference
rules.

For most purposes, big-step and small-step semantics are considered to
be equally expressive for describing terminating computations but not
for diverging (non-terminating) computations.
While both the step relation of small-step and the evaluation relation
of big-step semantics are inductively defined, small-step semantics
naturally formalizes diverging computations:
They are unbounded transition sequences that never reach a final
state.
Capturing diverging computations in the dynamics is important for
describing their effects and for distinguishing them from failing
computations (transition sequences that get stuck in a state that is
not final), which is crucial for formulating and proving type
soundness.
Big-step semantics, in its standard definition, only relates
terminating programs to their outcomes.
This means that it does not describe the effects of diverging
computations and cannot distinguish between diverging and failing
computations. Big-step semantics is therefore not suitable for some
purposes such as formulating type soundness,
which can be expressed elegantly with small-step semantics using
\emph{progress} and \emph{preservation}~\cite{wright1994syntactic}.

Despite this shortcoming, big-step semantics remains a popular choice
in the literature~\cite{leroy2009coinductive,chargueraud2013pretty}.
The use of big-step semantics in comprehensive and influential
works, such as the definition of Standard ML~\cite{MilnerTH1990} and
formalizations of C~\cite{Norrish99,BlazyL09}, is a
testament to its scalability, succinctness, and readability.
%
Many authors prefer big-step semantics because it can lead to simpler and
more succinct proofs.
Examples include correctness proofs for semantics-preserving
program translations arising in compiler verification 
\cite{leroy2006formal,klein2006machine,strecker2005compiler} 
as well as for type soundness proofs where type
judgments imply properties of global outcomes, such as resource
bounds~\cite{hofmann2003static,hoffmann2022two}.
One reason is that the structure of big-step semantics often
corresponds to the typing rules, and big-step semantics directly
describes outcomes without fussing with the intermediate program states reached
along the way.

There is a line of work that proposes to use \emph{coinductive} definitions
to remedy the deficiencies of big-step semantics for describing
divergence 
\cite{chargueraud2013pretty,cousot1992inductive,cousot2009bi,crole1998lectures,
dagnino2022meta,danielsson2012operational,hughes1995making,
leroy2009coinductive,
nakata2009trace,poulsen2017flag,zuniga2022coevaluation}.
Such work aims to generalize big-step semantics to enjoy its
advantages in situations where traditional big-step semantics is not
suitable.
Since the outcome of a diverging computation can be infinite, it is
natural to use coinduction to capture outcomes.
However, coinductive big-step semantics has some undesirable
properties and has not been widely adopted.
For one thing, proofs by induction on the semantic judgment turn into
coinductive proofs that hamper the close correspondence with
inductively defined type judgments.
For another thing, coinductive big-step semantics can yield 
counter-intuitive notions of evaluation (coevaluation)
or result in an explosion in the
number of inference rules \cite{leroy2009coinductive}.

In this 
work, we develop a little-known alternative
approach that extends big-step semantics with \emph{inductive}
rules to capture diverging computations.
Our work is based on an idea proposed by Hoffmann and Hofmann
that we dub \emph{big-stop semantics} \cite{hoffmann2010amortized}.
Instead of describing complete, possibly infinite outcomes using
coinduction, big-stop semantics inductively describes all finite
prefixes of outcomes. 
Moreover, big-stop semantics does so without significantly 
altering the language, whereas other inductive work 
adds features like
\emph{error states} 
and \emph{fuel} \cite{ernst2006virtual,siek2012,amin2017type} (see 
\Cref{sec:related}). 
Inductively describing finite prefixes of computation corresponds to
small-step semantics.

The essence of big-stop semantics is 
to take a standard formulation of
big-step semantics for a language 
and add just one rule schema for nondeterministically halting the computation.
In a well-chosen setting, this schema can even be reduced 
to just a single rule (see \Cref{sec:opt}).
Because big-stop semantics is just a small
extension to standard big-step semantics, it retains all the typical
benefits of big-step semantics.
However, because big-stop semantics can precisely match the relation
of small-step semantics, big-stop semantics is suitable for reasoning
about nonterminating computation.




This work makes the following contributions:

\begin{itemize}

  \item The big-stop system of operational semantics 
  is described for purposes of simplifying proofs 
  covering diverging computation. This system is created by 
  extending big-step with rules for nondeterministically 
  stopping computation. This extension is exemplified 
  for call-by-value PCF \cite{plotkin1977lcf} both with and without effects
  (\Cref{sec:stop,sec:eff}),
  as well as a while-loop-based imperative language (\Cref{sec:imp}), 
  showing that this 
  technique is broadly applicable.

  \item Formal properties of the big-stop system for 
  call-by-value PCF are proven 
  in Agda \cite{artifact} as marked by the symbol \agda. 
  These properties include the key equivalence of \Cref{thm:stseeq}
  showing that big-stop semantics coincides with the reflexive-transitive 
  closure of small-step semantics 
   and the progress property \Cref{thm:proge}.
   See \Cref{sec:eff} for further discussion.

  \item In \Cref{sec:kmach},
  Agda proofs (again marked by \agda) are given for the correctness of 
  compiling effectful call-by-value PCF to a Krivine machine \cite{krivine2007call}.
  These proofs exhibit the ease of big-stop formalization
  where using small-step semantics is unpleasant and using big-step 
  semantics is insufficient. 
  %
  They also exemplify a general principle---which we call the
  \emph{big-stop method}---for systematically extending a theorem
  about terminating computation that relies on big-step semantics to
  diverging computations.

  \item Further ergonomic optimizations to big-stop semantics 
  are outlined in \Cref{sec:opt}. These optimizations 
  allow big-step semantics to be extended to big-stop semantics 
  with as little as one additional rule, which can significantly 
  reduce the number of proof cases required when 
  performing rule induction.
  As described in \Cref{sec:related}, existing systems 
  covering the same niche as big-stop semantics
  typically require more complex extensions.

\end{itemize}

\section{Small-Step and Big-Step Semantics}\label{sec:background}

We start by providing the standard definitions of small-step and
big-step semantics.

\subsection{Call-By-Value PCF}

The language we consider throughout most of this article and in the
mechanization is a variant of PCF~\cite{plotkin1977lcf} with call-by-value (or
eager) evaluation order.
We choose PCF to be the basis of
our language because it is simple yet expressive enough to demonstrate
the most interesting properties of big-stop semantics.
Nonetheless, we do not believe 
that call-by-value evaluation is essential and expect
all
concepts to transfer to call-by-name without complications.%
\footnote{Launchbury has shown that call-by-need languages 
    also admit big-step semantics \cite{launchbury1993natural},
    which should be readily adaptable.
    %
    }
Note that we 
also consider untyped PCF expressions to make certain points
that are masked in the typed language.

\paragraph{Expressions and values}

The expression syntax of 
our version of pure PCF is given by the grammar of $e$ in \Cref{fig:syntax}
where $f,x$ range over variable names.
In the order appearing in the grammar,
expressions can be variables,
unary natural numbers, the
case analysis for natural numbers,
recursive functions, or
function applications.
%
The expression grammar is set up so that every expression 
in the language
is of the form $E[e_1,\dots,e_n]$, where $E$ indicates 
the kind of expression along with any variable bindings,
and where each $e_i$
is the $i^{th}$
subexpression that would need to be a value 
for the whole expression to be a redex. This notation makes it 
convenient to pick out the list of such subexpressions uniformly 
from any given expression. It should be noted 
that, despite superficial similarities,
this notation is \textit{not} the evaluation context 
notation of Felleisen \cite{felleisen1987calculi} (though 
we discuss evaluation contexts more in 
\Cref{sec:smallstep} and \Cref{sec:opt}).
To make this notation friendlier, we elide empty 
such lists so that, e.g., 
we write $\z$ instead of $\Z$.

\begin{figure}[t!b]
    \small
    \begin{subfigure}{0.4\textwidth}
    \[\begin{array}{lcl}
    e  & ::= & x \mid \Z \mid \suc {e'} \\
     && \mid \ifz {e_1} x {e_2} {e_3} \\
     && \mid \Lam f x e 
     \\
     &&\mid \app {e_1} {e_2} 
\end{array}\]
\caption{Expressions}
\label{fig:syntax}
\end{subfigure}
\hspace{4em}
    \begin{subfigure}{0.4\textwidth}
    \begin{mathpar}
    \inferrule[V-Zero]{
    }{
        \val \z
    }

    \inferrule[V-Succ]{
        \val v
    }{
        \val {\suc v}
    }

    \inferrule[V-Fun]{
    }{
        \val {\lam f x e}
    }
    \end{mathpar}
    \caption{Values}
\label{fig:values}    
\end{subfigure}
\hspace{2em}
\caption{Language of call-by-value PCF}
\label{fig:lang}
\end{figure}

As expected of PCF, this language extends lambda calculus with natural numbers
and a fixed-point operator. The expression $\z$ represents the number
0 and $\suc n$ represents $n+1$.  
The presence of primitive natural
numbers provides a simple way of writing (ill-typed)
programs that get
stuck, such as the application of zero to itself, $\app \z \z$.  Additionally,
the presence of a fixed-point operator allows one to write
nonterminating programs, even in a well-typed setting (for types, see
\Cref{fig:ty}).

Values $v$ of the language are indicated by the judgment 
$\val v$. The rules defining this judgment are given in
\Cref{fig:values}. The only well-typed values of the language 
are natural numbers and functions.

PCF typically evaluates its fixed-point operator lazily,
but we wish to work in an eager setting. There are at least 
two ways to work around this obstacle: either 
explicitly work with thunks, or change to 
a suitable eager fixed-point operator. 
While either solution should suffice,
our presentation uses the latter solution, 
effectively 
exchanging the Y fixed-point operator for 
the Z fixed-point operator. 
This exchange results in the 
expression $\lam f x e$,
which defines the recursive function $f$ 
with parameter $x$.
This fixed-point expression also happens to 
subsume the behaviour of function abstraction,
so we simplify by having it perform that role as well.
When abstracting a non-recursive function,
we may simply write $\lam {\_} x e$.

\paragraph{Static Semantics}

The static semantics of PCF is based on the type judgment
$\Gamma \vdash e : \tau$. It states that the expression $e$ has type
$\tau$ in context $\Gamma$, which assigns types to the free variables
in $e$. The definition is standard and given in \Cref{fig:ty} 
in \Cref{sec:types}.

The programs of PCF are defined to be well-typed, closed expressions,
that is, expressions $e$ such that $\cdot \vdash e : \tau$ for some
$\tau$ and the empty context $\cdot$.
In the mechanization of the results, we only consider well-typed
expressions.
However, many of our results also hold for untyped expressions and we
present these stronger results if applicable.

\subsection{Small-Step Semantics}\label{sec:smallstep}

Small-step semantics grew out of Gordon
Plotkin's structural operational 
semantics (SOS) \cite{plotkin1981structural}.
The idea behind this approach is to describe individual computation
steps in a transition system in a structurally-directed manner.
Describing operational
semantics in this way allows for program behaviour to be reasoned
about inductively in a straightforward, syntax-oriented way.

\begin{figure}[t!]
    \small
    \begin{mathpar}
        \inferrule[S-Seq(k)]{
            1 \leq k \leq n
            \\
            \forall\, 1\leq i < k.\, \val {e_i}
            \\
            \ssj {e_k} {e_k'}
            \\
            \forall\, i \neq k.\, e_i = e_i'
        }{
            \ssj {E[e_1,\dots,e_n]} {E[e_1',\dots,e_n']} 
        }

        \inferrule[S-CaseZ]{
        }{
            \ssj {\ifz {e_1} {x} {e_2} \z} {e_1} 
        }

        \inferrule[S-CaseS]{
            \val v
        }{
            \ssj {\ifz  {e_1} {x} {e_2} {\suc v}} {[v/x]e_2} 
        }

        \inferrule[S-App]{
            \val v
        }{
            \ssj {\app {\lam {f} {x} {e}} v} {[\lam {f} {x} {e}/f,\, v/x]e} 
        }
    \end{mathpar}
    \caption{Small-step transitions for pure call-by-value PCF}
    \label{fig:sss}
  \end{figure}

Pure PCF's call-by-value small-step transitions 
are given in \Cref{fig:sss}, where the judgment
$\ssj {e_1} {e_2}$ means that 
the expression $e_1$ transitions to 
the expression $e_2$ in one computational step.
We then take the reflexive-transitive closure 
of this small-step relation, $\msj {e_1} {e_2}$,
to describe sequences of computation as defined in \Cref{fig:mss}.

\begin{wrapfigure}{RI}{5.5cm}
  \small
  \centering{
  \begin{minipage}[c]{5cm}
    \begin{mathpar}
        \inferrule[M-Refl]{
        }{
            \msj {e} {e} 
        }

        \inferrule[M-Step]{
            \ssj {e_1} {e_2} 
            \\
            \msj {e_2} {e_3} 
        }{
            \msj {e_1} {e_3} 
        }
    \end{mathpar}
  \end{minipage}}
    \caption{Multi-step semantics}
    \label{fig:mss}
\end{wrapfigure}

Small-step semantics typically require multiple rules 
for language constructs with multiple 
subexpressions,
which are called \textit{congruence rules}. 
Such rules can bloat the number of rules 
required for small-step semantics.
However,
the rules of \Cref{fig:sss} avoid listing out all 
the congruence rules
through the use of the rule schema \textsc{S-Seq(k)}. 
For each choice of the parameter $k$, this schema results in a
rule that says that one computation step 
can be made, in place, for the $k^{th}$ subexpression
of $E[e_1,\dots,e_n]$
so long as all
prior subexpressions are already values.
This setup also enforces left-to-right 
evaluation of the subexpressions.
For 
our particular language, this rule schema replaces 
four congruence rules: one each to step the function
and argument in an application, one to step the 
contents of a successor, and one to step the argument 
to a case analysis.

An alternative to using our rule schema is to use evaluation contexts
\cite{felleisen1987calculi}.
Evaluation contexts provide a meta-syntax for expressions that helps
deal with congruences by cutting out a hole to indicate where an
expression is to step next.
In a small-step semantics with evaluation contexts, we would replace
the rule schema \textsc{S-Seq(k)} with one rule stating that
expressions in holes can step independently of their surrounding
context.
However, this approach inherently orients the language toward
small-step semantics, as big-step semantics considers more than just
the next step of evaluation.
Thus, to prevent obscuring eventual comparison between big-step 
and big-stop
semantics, we initially focus on schema-based semantics and
revisit evaluation contexts in \Cref{sec:opt}.

Aside from \textsc{S-Seq(k)}, 
the remaining three rules of \Cref{fig:sss}
are simply the rules for redexes in our language.
\textsc{S-CaseZ} and \textsc{S-CaseS} 
are the case analysis rules for zero and non-zero 
numbers, respectively, and 
\textsc{S-App} is the rule for function application.

Small-step systems of operational semantics 
are used to state some foundational results in 
the study of programming languages. 
For example, small-step semantics 
are used to state the concepts of \textit{progress} and 
\textit{preservation}, which are the key underpinnings 
of syntactic proofs of type soundness \cite{wright1994syntactic}.
``Progress'' states that well-typed 
expressions either are values or can take a small step.
``Preservation'' states that the typing of an expression is 
maintained through taking small steps.
If both progress and preservation hold for a language,
then well-typed 
expressions cannot get stuck at
type errors during evaluation. This approach 
does not require that evaluation terminates.


\iffalse
\subsubsection{Effects}

Our motivating problem also considers effects. Specifcally, it considers
writer-monad style effects that are
emitted by computation but do not affect computation. Such effects
include behaviour like printing text and accumulating cost but do not
include behaviour like changing a mutable state.

To capture writer-monad-style effects,
we consider the traces of effects in the order 
they occur.
Thus, for example, the effect term $abc$ just means that 
effect $a$ is followed by effect $b$ and then effect $c$.
Such effects are emitted
by the expression $\eff a e$, which has the effect 
indicated by $a$ and continues on to evaluate as 
the expression $e$.

Such effects can be incorporated into the small-step 
semantics using the judgment $\ssej {e_1} {e_2} a$,
meaning that $e_1$ small-steps to $e_2$ and 
has effect $a$. The rules for this system 
are given in \Cref{fig:sses}.
Here, we consider $1$ to be the identity effect, which can be thought
of as an uninteresting effect or no effect at all.

\Cref{fig:sses} also defines a multi-step judgment
$\msej {e_1} {e_3} {s}$, where $s$ is a sequence of effects.
Such a sequence can be treated as a
\textit{monoid} in which $1$ is the identity element.

\begin{figure}
    \begin{mathpar}
        \inferrule[SE-Seq(k)]{
            1 \leq k \leq n
            \\
            \forall\, 1\leq i < k.\, \val {e_i}
            \\
            \ssej {e_k} {e_k'} a
            \\
            \forall\, i \neq k.\, e_i = e_i'
        }{
            \ssej {E[e_1,\dots,e_n]} {E[e_1',\dots,e_n']}  a
        }

        \inferrule[SE-CaseZ]{
        }{
            \ssej {\ifz {e_1} {x} {e_2} \z} {e_1} 1
        }

        \inferrule[SE-CaseS]{
            \val v
        }{
            \ssej {\ifz {e_1} {x} {e_2} {\suc v}} {[v/x]e_2} 1
        }

        \inferrule[SE-App]{
            \val v
        }{
            \ssej {\app {\lam {f} {x} {e}} v} {[\lam {f} {x} {e}/f,\, v/x]e} 1
        }

        \inferrule[SE-Effect]{
        }{
            \ssej {\eff a e} e a
        }
        \\
        \inferrule[ME-Refl]{
        }{
            \msej {e} {e} 1
        }

        \inferrule[ME-Step]{
            \ssej {e_1} {e_2} a 
            \\
            \msej {e_2} {e_3} b
        }{
            \msej {e_1} {e_3} {ab}
        }
    \end{mathpar}
    \caption{Small- and multi-step semantics for effectful call-by-value PCF}
    label{fig:sses}
    \end{figure}
\fi

\iffalse 
\subsubsection{Types}
\label{sec:ty}

\begin{figure}
    \begin{mathpar}
        %
        %
        \inferrule[T-Var]{
        }{
            \ty {\Gamma,x:\tau} x \tau
        }

        \inferrule[T-Lam]{
            \ty {\Gamma, f:\fun \tau \sigma, x:\tau} e \sigma
        }{
            \ty {\Gamma} {\lam {f} {x} e} {\fun \tau \sigma}
        }

        \inferrule[T-Zero]{
        }{
            \ty \Gamma \z \nat
        }

        \inferrule[T-Succ]{
            \ty \Gamma e \nat
        }{
            \ty \Gamma {\suc e} {\nat}
        }

        \inferrule[T-App]{
            \ty \Gamma {e_1} {\fun \sigma \tau} 
            \\
            \ty \Gamma {e_2} {\sigma}
        }{
            \ty \Gamma {\app {e_1} {e_2}} \tau
        }

        \inferrule[T-Case]{ 
            \ty \Gamma {e_3} \nat
            \\
            \ty \Gamma {e_1} \tau 
            \\
            \ty {\Gamma, x:\nat} {e_2} \tau
        }{
            \ty \Gamma {\ifz {e_1} x {e_2} {e_3}} \tau 
        }

        \inferrule[T-Eff]{
            \ty {\Gamma} e \tau
        }{
            \ty {\Gamma} {\eff a e} \tau
        }
    \end{mathpar}
    \caption{Typing rules for PCF}
\label{fig:ty}
\end{figure}

To properly consider well-typed settings,
we lay out the standard type system of PCF. The only types of PCF are
natural numbers and functions, formalized by the following grammar:
\[ \tau ::= \nat \mid \fun {\tau_1} {\tau_2} \]

The typing rules of \Cref{fig:ty} assigns 
these types to expressions using the
judgment $\ty{\Gamma}{e}{\tau}$, which
means that the expression $e$ has type $\tau$ 
given the typing assumptions of the typing context $\Gamma$.
The rules are standard, with the exception of \textsc{T-Eff}. It
assign the type $\tau$ to the expression $\eff a e$ in context
$\Gamma$ if $e$ has type $\tau$ in $\Gamma$.
Note that types are not unique since functions can be assigned
different argument types.
\fi

\subsection{Big-Step Semantics}

Big-step semantics originates 
from Gilles Kahn's natural semantics \cite{kahn1987natural}.
In a sense,
the big-step approach achieves
the main thrust of small-step semantics by 
directly relating expressions to the values they reduce to.
The big-step relation is 
written here as $\bsj {e} {v}$, meaning that the 
expression $e$ evaluates to the value $v$.
The rules of big-step semantics are syntax-directed and given in
\Cref{fig:bigpure}.

\begin{figure}
    \small
\begin{mathpar}
 \inferrule[B-Val]{
       \val v
    }{
        \bsj {v} {v} 
    }

  \inferrule[B-Succ]{
    \bsj e v
  }{
    \bsj {\suc e} {\suc v}
  }

  \inferrule[B-CaseZ]{ 
    \bsj e \z 
    \\ 
    \bsj {e_1} {v_1} 
  }{ 
    \bsj  {\ifz {e_1} x {e_2} e} {v_1} 
  }
    
  \inferrule[B-CaseS]{ 
    \bsj e {\suc v} 
    \\
    \bsj {[v/x]e_2} {{v_2}} 
  }{ 
    \bsj  {\ifz  {e_1} x {e_2} e} {v_2}
  }
    
  \inferrule[B-App]{ 
    \bsj {e_1} {\lam f x e} 
    \\ 
    \bsj {e_2} {v_2} 
    \\
    \bsj {[{\lam f x e}/f, v_2/x]e} {v} 
  }
  { \bsj {\app {e_1} {e_2}} {v}}
\end{mathpar}
\caption{Big-step semantics for pure call-by-value PCF}
\label{fig:bigpure}
\end{figure}

Crucially, these big-step semantics match those given 
by the small-step system of \Cref{fig:sss}.
This equivalence can be formalized by \Cref{lem:bseq}, which can be
proved by straightforward rule induction. More precisely, this equivalence
states that the two systems agree on terminating computations.
However, this is already sufficient to obtain equivalence of the two
semantics for a type-safe language without effects such as pure PCF
since divergence is the only other possible outcome.

\begin{theorem}[Big/Multi Equivalence]\label{lem:bseq}
    For all expressions $e$ and values $v$,
    $\bsj e v \iff \msj e v$.
\end{theorem}





Big-step rules are very similar in structure to 
sequent-style rules for natural deduction.\footnote{Hence 
the name ``natural semantics''.} Because 
typical typing rules are also structurally 
similar to these rules, there is a good correspondence
between a language's types and big-step operational semantics.
%
This high level of structural similarity
between semantics and 
typing rules comes with ergonomic benefits.
For example, proving type soundness for these rules is
rather straightforward because the structures 
of the derivations of corresponding judgments match.
In contrast, the small-step 
system's congruence rules typically have
no structural match to typing rules, 
rendering a more cumbersome type soundness proof
that might require, e.g., context substitution lemmas.

Another benefit of using big-step semantics
is that intermediate computations never need to be inspected.
This property manifests in how the application of a given
big-step rule depends solely on the identity and subexpressions
of the expression $e$ in $\bsj e v$; the value $v$ is irrelevant
for syntax-guided rule application. In contrast, rule 
\textsc{M-Step} from \Cref{fig:mss} 
has expression $e_2$ as the righthand 
element of the premiss $\ssj {e_1} {e_2}$
\textit{and} the left-hand element of the premiss $\msj {e_2}{e_3}$.
As a result, every intermediate 
computation $e_2$ must be inspected when using multi-step semantics.

However, the benefits of big-step semantics 
come at a cost: if an expression does not 
evaluate, big-step semantics cannot reason about it. 
As a result, programs that loop forever, get 
stuck on undefined behavior, etc. cannot 
be described by big-step semantics. 
In contrast, small-step semantics can 
describe execution arbitrarily 
deep into infinite loops and
right up to getting stuck.


\section{Motivation: Semantic Preservation}\label{sec:motive}

It has been argued in the
literature~\cite{leroy2009coinductive,chargueraud2013pretty} that
big-step semantics can simplify compiler correctness proofs.
When the dynamic semantics of the source language is different from
the (lower-level) dynamic semantics of the target language, it can be
hard for a proof of semantic preservation to proceed by relating each
language's intermediate states of computation, as they may not share
any convenient structural similarities.  Big-step semantics can avoid
these problems by completely skipping past intermediate states of
computation.
%

Here we exhibit a particularly elementary setting where this problem
still arises: semantic preservation for call-by-value PCF with respect
to a Krivine machine (K machine)
\cite{krivine2007call,douence2007next,leroy1990zinc}.
If we view this problem through the lens of compiler correctness, then
the source and target languages are identical, so the compiler should simply 
copy the source unchanged into the K machine.
However, this case already demonstrates that typical proofs of semantic
preservation for compilation with respect to PCF's small-step
semantics are inconvenient.  This section describes (1) where such
inconveniences arise, (2) how big-step semantics can be used to
alleviate them, and (3) why it is desirable to
generalize this approach to diverging computations using big-stop
semantics.

\paragraph{The K Machine}

A Krivine machine (K machine) is a stack-based, abstract machine
similar to the SECD machine \cite{landin1964mechanical}.
While Krivine \cite{krivine2007call}
designed his machine for call-by-name, 
the concept is quite flexible and has been 
adapted here for our call-by-value PCF. Our presentation 
of the K machine follows the notational conventions of 
\emph{Practical Foundations of Programming Languages} (PFPL) \cite{harper2016practical} (Chapter 28),
which also considers the correctness 
of a K machine with respect to SOS.

Whereas small-step systems 
need various 
congruence rules for propagating computation into 
subexpressions (e.g., \textsc{S-Seq(k)} and
\textsc{SE-Seq(k)}),
a K machine deals with such congruences directly in its 
term language by working with
an explicit stack. The K machine operates on an expression 
by pushing subexpressions to be evaluated onto the stack 
and popping values from the stack.

Each state in the K machine 
takes the form of either $\estate{k}{e}$ or 
$\vstate k e$ where $k$ is a stack of 
\textit{frames} and $e$ is a 
PCF expression. States of the form $\estate{k}{e}$
indicate that the expression $e$ is not yet a value 
and needs to be evaluated, and states of the form $\vstate k v$
indicate that $v$ has been found to be a value and 
is ready to be plugged into the top frame of the stack $k$.
These features of the K machine are
formally described in \Cref{sec:kmach} 
across \Cref{fig:frames,fig:ke}.
In particular, \Cref{fig:ke} defines the judgment 
$\ksj S T$ to indicate that state $S$ transitions to 
state $T$, in addition to its multi-step analogue
$\kmj S T$.
For example, the following rules define the transitions for function
application.
\begin{mathpar}
  \small
  \inferrule[K-App1]{ 
  }{ 
    \ksj {\estate k {\app {e_1} {e_2}}} {} \\\\ {\estate {k;\ffun {e_2}} {e_1}} 
  }

  \inferrule[K-App2]{ 
  }{ 
    \ksj {\vstate {k;\ffun {e}} {v}} {}\\\\ {\estate {k;\farg {v}} {e}} 
  }

  \inferrule[K-App3]{ 
  }{ 
    \ksj {\vstate {k;\farg {\lam f x e}} {v}}  {} \\\\ {\estate k {[\lam f x e/f, v/x]e}} 
  }
\end{mathpar}

\subsection{Proving the Correctness of the K machine}

Now we can discuss compiling PCF to the K machine and semantic
preservation.
Compiling a PCF expression $e$ to the K machine simply yields the
state $\estate \epsilon e$. To verify that this (trivial) compilation
preserves the semantics of expressions, the following two key lemmas
are desired.

\begin{lemma}[Soundness]
  \label[lemma]{lem:sound}
  If $\kmj {\estate \epsilon e} {\vstate \epsilon v} $ then $\msj e v $.
\end{lemma}

\begin{lemma}[Completeness]
  \label[lemma]{lem:compl}
  If $\msj e v $ and $\val v$ then $\kmj {\estate \epsilon e} {\vstate \epsilon v} $.
\end{lemma}

Even though these
lemmas are essentially by-the-book, 
they are already inconvenient to prove,
and we largely follow
PFPL \cite{harper2016practical} to demonstrate this fact.
The problem is that there does not exist a direct proof by rule
induction on the judgments on the left side of the
implication.
Instead, the classic proof approaches require the formal development
of additional theoretical machinery, such as simulation relations or
alternative semantics, in particular for the proof of
\Cref{lem:compl}, which we focus on in this section.
To summarize the results that we find: proofs using big-step semantics
can be carried out by simple rule induction and are easier than other
approaches.

\paragraph{Proving Completeness Directly}

The direct proof of \Cref{lem:compl} (completeness) would 
induct over the
derivation of the antecedent.
So consider proving completeness by rule
induction over the derivation of the multi-step relation $\msj {e} {v}$.
First, it is necessary to generalize the inductive hypothesis
to consider any stack $k$:
If $\msj e v$ and $\val v$ then $\kmj {\estate k e} {\vstate k v}$

Now consider the structure of the multi-step relation's derivation.
One possibility is that 
the relation was derived via \textsc{M-Step}, 
requiring the premisses $\ssj e {e'}$ and 
$\msj {e'} v$ for some expression $e'$.
The latter premiss is ripe for applying the inductive 
hypothesis, but the former is not because it does not use the 
multi-step judgment and $e'$ may not 
be a value. Thus, further reasoning is warranted.

To continue the inductive 
proof, one needs to show the following statement, where 
the value $v$ is the
evaluation of both expressions $e$ and $e'$:
If $\ssj e {e'}$ and $\kmj {\estate k {e'}} {\vstate k v} $, then
$\kmj {\estate k e} {\vstate k v} $.

Given the antecedent,
by transitivity, it would suffice to 
show that $\kmj {\estate k e} {\estate k e'}$.
Similarly, by determinacy, it would suffice to show 
that $\kmj {\estate k {e'}} {\estate k e}$.
However, neither statement holds, as the 
next time the stack $k$ should be present after transitioning from 
a state of the form
$\estate k {e''}$ should be in the state $\vstate k {v''}$ where 
$v''$ is the evaluation of $e''$. 
Thus, neither of the transition sequences of $\estate k {e}$ nor $\estate k {e'}$
should be expected to be an
extension of the other, and the
first point at which their transition sequences
coincide should be expected to be $\vstate k v$.
This circumstance is structurally inconvenient, and 
one is forced to continue 
by inner induction on the structure of $\ssj e {e'}$.

Many cases here pose a problem. Consider
the case from
\textsc{S-Seq(1)} where 
$e = \app {e_1} {e_2}$ and $e' = \app {e_1'} {e_2}$ for some expressions 
$e_1, e_1', e_2$ where $\ssj {e_1} {e_1'}$. 
Here one wants to show that, if 
$\kmj {\estate k {\app {e_1'} {e_2}}} {\vstate k v}$, 
then $\kmj{\estate k {\app {e_1} {e_2}}} {\vstate k v}$.
To continue this case, note that \textsc{K-App1} yields both that
$\ksj {\estate k {\app {e_1} {e_2}}} {\estate {k;\ffun {e_2}} {e_1}}$
and that $\ksj {\estate k {\app {e_1'} {e_2}}} {\estate {k;\ffun {e_2}} {e_1'}}$,
which unwraps the first steps of both transition sequences of interest.
At this point, the proof seems feasible if 
one can apply an inductive hypothesis.
However, both the outer and inner 
inductive hypotheses would 
require the judgment $\msj {e_1'} {v_1}$ 
(at least as a subderivation of $\msj {e_1} {v_1}$),
but this judgment is not readily 
available and it is not clear how to proceed.
%


\paragraph{The Simulation Approach}

Simulation arguments are an available avenue to proving 
\Cref{lem:compl} if one 
is willing to develop some additional 
proof machinery.
Specifically, a simulation requires defining a 
functional
relation $\looparrowright$
between K machine states and PCF expressions such 
that:
\begin{tasks}[style=itemize](2)
  \task $\estate \epsilon e \looparrowright e$
\task $\vstate \epsilon v \looparrowright v$
\end{tasks}
\begin{tasks}[style=itemize]
\task if $\msj e {e'}$ and $S \looparrowright e$, then 
there is some state $T$ such that
$T \looparrowright e'$,
and $\kmj S T$
\end{tasks}

A sensible such simulation relates
each K machine state to the PCF 
expression obtained by unwinding its stack. 
For example, $\estate \fsucc \z \looparrowright \suc \z$. 
This relation clearly 
satisfies the first two necessary properties.
However, verifying the last of the necessary properties 
for this relation is a problem. The issue can be seen 
when considering the case where $e'$ is a value 
and $S$ is chosen to be $\estate \epsilon e$. 
Then this final property is just a slight generalization 
of the completeness lemma itself. Proving this slight generalization 
is no easier.

Another conceptual annoyance for making a simulation argument is 
that there is no bound on how many K machine states may be related 
to an expression, nor is there a bound on how many steps 
of the K machine may be necessary to match a single step in PCF.
As a result, one must at the very 
least perform an inner induction on the structure of 
the expression $e$, which runs into similar 
issues as the basic inductive approach.

Note that the complexity of a simulation argument would increase 
when considering language features like effects, different
source and target languages, and/or more involved compiler passes.

\paragraph{The Big-Step Approach}

The textbook approach to proving \Cref{lem:compl} 
is to avoid using small-step semantics entirely.
First one derives \Cref{lem:bseq} (big/multi equivalence), which 
follows from straightforward inductive arguments.
Then
\Cref{lem:compl} is a consequence of the
following lemma,
which can be directly proved by induction on the judgment $\bsj e v$.

\begin{lemma}[Big-Step Completeness]
\label[lemma]{lem:bigcomplete}
  If $\bsj e v$ then $\kmj {\estate k e} {\vstate k v}$.
\end{lemma}

The key to the proof of \Cref{lem:bigcomplete} is that the
big-step derivation tracks a resulting value $v$ for
every relevant subexpression of $e$. With such a value always handy,
the recovery of such evaluations is trivial, and big-step semantics
avoids the problem encountered by small-step semantics.

\subsection{Accounting for Diverging Computations}

If one assumes that programs are well-typed, then it is sufficient for
\Cref{lem:sound,lem:compl} to focus on computations 
that result values.
Thanks to type-safety, computations cannot get stuck,
so the only two possible outcomes for computations 
are to converge to a value or to
diverge and loop forever. Because all divergence is observationally 
identical in a pure language, one only needs to ensure 
that every diverging expression in PCF is compiled to some diverging 
state in the K machine. This property clearly 
holds because \Cref{lem:sound,lem:compl} clearly show the 
contrapositive---converging expressions are compiled to 
converging states.

However, the story changes 
if one considers effects, which real-world languages virtually
always have. In a setting 
with effects, diverging computations can be distinguished
as they may induce different effects. These distinctions arise
even with very simple writer-monad-style effects that do not 
affect computation, like printing. To compile such 
a language correctly, it is necessary to ensure 
that the observable effects emitted by the source language and 
target language match, regardless of program termination.

In \Cref{sec:eff,sec:kmach}, we extend both SOS and 
the K machine with
effects, defining the judgments
\[
  \begin{array}{ccc}
    \msej {e_1} {e_2} a & \hspace{2em} \text{and} \hspace{2em} & \kemj{S}{T}{a}
  \end{array}
\]
where $e_i$ are expressions, $S$ and $T$ are K machine states,
and $a$ is an abstract sequence of effects.
%

In a setting with effects, additional properties 
analogous to \Cref{lem:sound,lem:compl}
must be formalized to handle diverging computations.
In particular, our goal is to prove \Cref{lem:soundd,lem:compld}.
Essentially, these properties state that, if 
one dynamic semantics emits some sequence of effects at some point, 
the other semantics emits a matching sequence 
of effects.
Without these latter two lemmas, programs that 
run indefinitely cannot be considered to have preserved semantics in the K machine.

\begin{lemma}[Divergent Soundness]
  \label[lemma]{lem:soundd}
  If $\kemj {\estate \epsilon e} {S} a$ then $\msej e {e'} a$ for some expression $e'$.
\end{lemma}

\begin{lemma}[Divergent Completeness]
  \label[lemma]{lem:compld}
  If $\msej e {e'} a$ then $\kemj {\estate \epsilon e} {S} a$ for some 
  state $S$.
\end{lemma}

Unfortunately, while big-step semantics are the nicest way 
of showing \Cref{lem:compl}, they have no hope of showing 
the diverging analogue, as big-step semantics can only 
describe converging computations, not diverging ones.
It therefore appears that one must fall back to less favorable 
approaches like simulation.

For this reason, this article develops 
big-\textit{stop} semantics, which can maintain 
the niceties of
big-step-style reasoning even for diverging computations.
Specifically, we show in \Cref{sec:kmach} that
\Cref{lem:compld} can be proven with straightforward rule induction.
The key step is formulating and proving the big-stop analog of
\Cref{lem:bigcomplete} to account for diverging computations.
Alternative approaches for handling divergence are 
discussed in \Cref{sec:related} and compared to big-stop.


\section{Pure Big-Stopping}\label{sec:stop}

This section develops big-stop semantics to 
conveniently reason about potentially nonterminating 
programs while maintaining the ease of using big-step 
semantics. Moreover, this development proceeds in a minimally 
invasive manner, making only a small adaptation to the rules of
big-step semantics.
The essence of the approach is to take a standard big-step system 
and add just one rule schema 
to allow for arbitrarily stopping computation. 
The remainder of this section shows how to perform such 
a big-stop extension to our effectful, call-by-value PCF.
However, nothing about the technique is intrinsically 
tied to PCF, so the same technique should apply to 
other languages. We exhibit such an extension 
for an imperative language in \Cref{sec:imp}.

\iffalse
\subsection{Big-Step}

\begin{figure}
\begin{mathpar}
 \inferrule[B-Val]{
       \val v
    }{
        \bsj {v} {v} 
    }

  \inferrule[B-CaseZ]{ 
    \bsj e \z 
    \\ 
    \bsj {e_1} {v_1} 
  }{ 
    \bsj  {\ifz {e_1} x {e_2} e} {v_1} 
  }
    
  \inferrule[B-CaseS]{ 
    \bsj e {\suc v} 
    \\ 
    \bsj {[v/x]e_2} {{v_2}} 
  }{ 
    \bsj  {\ifz  {e_1} x {e_2} e} {v_2}
  }
    
  \inferrule[B-App]{ 
    \bsj {e_1} {\lam f x e} 
    \\ 
    \bsj {e_2} {v_2} 
    \\
    \bsj {[{\lam f x e}/f, v_2/x]e} {v} 
  }
  { \bsj {\app {e_1} {e_2}} {v}}
\end{mathpar}
\caption{Big-step semantics for pure call-by-value PCF}
\label{fig:bigpure}
\end{figure}

\begin{figure}
\begin{mathpar}
 \inferrule[BE-Val]{
       \val v
    }{
        \bsej {v} {v} 1
    }

  \inferrule[BE-CaseZ]{ 
    \bsej e \z a
    \\ 
    \bsej {e_1} {v_1} b
  }{ 
    \bsej  {\ifz {e_1} x {e_2} e} {v_1} {ab}
  }
    
  \inferrule[BE-CaseS]{ 
    \bsej e {\suc v} a
    \\ 
    \bsej {[v/x]e_2} {{v_2}} b
  }{ 
    \bsej  {\ifz {e_1} x {e_2} e} {v_2} {ab}
  }
    
  \inferrule[BE-App]
  { \bsej {e_1} {\lam f x e} a
    \\ 
    \bsej {e_2} {v_2} b
    \\ 
    \bsej {[{\lam f x e}/f, v_2/x]e} {v} c
  }{ 
    \bsej {\app {e_1} {e_2}} {v} {abc}
  }

  \inferrule[BE-Eff]{
    \bsej e v b
  }{
    \bsej {\eff a e} v {ab}
  }
\end{mathpar}
\caption{Big-step semantics for effectful call-by-value PCF}
\label{fig:bigeffect}
\end{figure}

Before we can extend to big-stop, we first need a big-step system.
Such a system for our pure call-by-value PCF is given by the rules of
\Cref{fig:bigpure}. These rules define the judgment $\bsj e v$ 
which means that expression $e$ evaluates to value $v$.

Crucially, these big-step semantics match those given 
by the small-step system of \Cref{fig:sss}.
This equivalence can be formalized by \Cref{lem:bseq},
which is a classic lemma proved in many undergraduate 
programming language classes. Really, this equivalence 
states that the two systems agree on terminating computations,
which is usually all that 
is needed in a type-safe language without effects.

\begin{lemma}[big/small pure equivalence]\label{lem:bseq}
    For all expressions $e$ and values $v$,
    $\bsj e v \iff \msj e v$
\end{lemma}

The effectful analogue of this big-step system is 
given in \Cref{fig:bigeffect}. 
This system uses the judgment $\bsej e v a$ to 
mean that expression $e$ evaluates to value $v$
while inducing the effects indicated by $a$.

Like the pure system, this system also has 
an equivalence with its small-step analogue.
This equivalence is formalized by \Cref{lem:bseeq}.
Once again, this equivalence only applies to terminating 
computations.

\begin{lemma}[big/small effectful equivalence]\label{lem:bseeq}
    For all expressions $e$, values $v$, and effects $a$,
    $\bsej e v a \iff \msej e v a$
\end{lemma}

While the proof of \Cref{lem:bseeq} is not particularly 
different from the classic proof of \Cref{lem:bseq},
we provide an Agda proof of the statement
in a setting with types. \david{how should this be cited?}
\fi

\begin{figure}
  \small
  \def \MathparLineskip {\lineskip=0.3cm}
  \begin{mathpar}
 \inferrule[St-Stop(k)]{
        \forall 1 \leq i \leq k. \,
        \bstj {e_{i}} {e'_{i}} 
        \\
        \forall 1 \leq i < k. \,
        \val {e'_i}
        \\
        \forall k+1 \leq i \leq n. \,
        e'_i = e_i
    }{
        \bstj {E[e_1, \dots, e_n]} {E[e'_1, \dots, e'_n]} 
    }

  \inferrule[St-CaseZ]{ 
    \bstj e \z 
    \\ 
    \bstj {e_1} {e'_1} 
  }{ 
    \bstj  {\ifz {e_1} x {e_2} e} {e'_1} 
  }

  \inferrule[St-CaseS]{ 
    \bstj e {\suc v} 
    \\ 
    \val v
    \\ 
    \bstj {[v/x]e_2} {{e'_2}} 
  }{ 
    \bstj  {\ifz {e_1} x {e_2} e} {e'_2} 
  }
    
  \inferrule[St-App]{ 
    \bstj {e_1} {\lam f x e} 
    \\ 
    \bstj {e_2} {v_2} 
    \\\\
    \val {v_2}
    \\ 
    \bstj {[{\lam f x e}/f, v_2 /x]e_2} {e'} 
  }{ 
    \bstj {\app {e_1} {e_2}} {e'} 
  }
\end{mathpar}
\vspace{-2ex}
\caption{Big-stop semantics for pure call-by-value PCF}
\vspace{-1.5ex}
\label{fig:stoppure}
\end{figure}

The fundamental idea of big-stop semantics is to add new 
rules to big-step semantics that
enable the semantics to nondeterministically stop a computation 
midway through. Each of the possible resulting
partially-computed expressions corresponds
to computing only a prefix of the complete sequence of
computation steps. Even if the complete 
computation sequence is infinite because the 
computation is nonterminating, its prefixes 
will be finite and inductively capturable. 
Capturing all
such prefixes of computation is sufficient to 
describe both terminating and nonterminating 
computations uniformly---terminating 
computations will simply have finitely many prefixes,
one of which ends with a value.
These prefixes also allow big-stop semantics 
to rival its coinductive alternatives
in reasoning about infinite computations,
in the same way that 
infinite streams can be equivalently reasoned about
directly using coinduction
or via their prefixes using induction.

When adapting a big-step system to big-stop,
it may also be necessary to 
add explicit premisses to existing big-step rules 
to ensure certain expressions 
are values. In the original big-step system,
such expressions would necessarily be values,
and thus such premisses are optional in the big-step system, while 
they are required in the big-stop extension. 
The details of this extension are exemplified 
for pure PCF in the following paragraphs.
Effectful PCF is covered in \Cref{sec:eff}, and 
\Cref{sec:opt} shows techniques that further minimize
the changes required for
this kind of extension.

Formally, the big-stop relation is given
by the judgment $\bstj e {e'}$, which means 
that the expression $e$ partially evaluates to the 
expression $e'$. 
The rules for this judgment are given in \Cref{fig:stoppure} 
and described in the following paragraphs.
Most of the rules correspond to \Cref{fig:bigpure}'s
big-step rules. Some examples of the use of
the big-stop rules can be found 
in \Cref{sec:examples}.

The key new feature in \Cref{fig:stoppure} is 
the rule schema
\textsc{St-Stop(k)}, which represents 
the core idea of nondeterministically halting 
evaluation. This rule allows any expression 
$E[e_1, \dots, e_n]$ to stop evaluating
after reducing only some of its 
subexpressions $e_i$. In accordance with 
the evaluation order of PCF, these subexpressions 
must be reduced from left to right
so that $e_i$ is reduced to a value when 
$i < k$ and is left untouched when $i > k$---only $e_k$ 
may be not fully evaluated.
Moreover, when instantiated at $k=0$, this rule schema 
allows all expressions to immediately stop evaluation 
by big-stopping to themselves, 
subsuming the role of 
\textsc{B-Val}, which already halted computation for values.
Finally, this rule schema is suggestively similar 
to the \textit{small}-step rule schema 
\textsc{S-Seq(k)} from \Cref{fig:sss},
hinting at the role it plays for big-stop 
semantics. 

To be explicit, the schema \textsc{St-Stop(k)} stands in for the following 
five rules:
\begin{small}
  \def \MathparLineskip {\lineskip=0.3cm}  
\begin{mathpar}
\inferrule[St-Stop]{
}{
  \bstj e e
}

\inferrule[St-Succ]{
  \bstj {e} {e'}
}{
  \bstj {\suc e} {\suc e'}
}

\inferrule[St-Case]{
  \bstj {e} {e'}
}{
  \bstj {\ifz {e_1} x {e_2} {e}} {\ifz {e_1} x {e_2} {e'}}
}
\\
\inferrule[St-App1]{
  \bstj {e_1} {e_1'}
}{
  \bstj {\app {e_1} {e_2}}  {\app {e_1'} {e_2}}
}

\inferrule[St-App2]{
  \bstj {e_1} {v_1}
  \\
  \val {v_1}
  \\
  \bstj {e_2} {e_2'}
}{
  \bstj {\app {e_1} {e_2}}  {\app {v_1} {e_2'}}
}
\end{mathpar}
\end{small}

Aside from \textsc{St-Stop(k)}, 
the other rules are nearly unchanged from their big-step 
counterparts. However, some additional premisses are added 
to ensure that certain subexpressions are values. 
For example, in \textsc{St-CaseS}, the premiss 
$\val v$ is added to ensure that $e$ fully 
evaluates to $\suc v$. This additional premiss is necessary to 
ensure that partial evaluations of $e$ are not prematurely 
plugged into the case expression. In big-\textit{step} semantics,
such a premiss is unnecessary because the dynamics
already ensure that $\suc v$ is a value.
As big-\textit{stop} semantics allows for partial 
evaluations, this invariant is no longer present.
Thus the new value premisses do not comprise 
new restrictions for the semantic rules, 
but rather explicitly maintain pre-existing restrictions. 
An alternative formulation of big-stop semantics could
simply use a big-step judgment for these premisses 
to maintain the same invariant, but it requires fewer rules 
overall to have only one kind of judgment.

One subtlety of the big-stop extension is 
that, while a big-step system may be deterministic, 
big-stop is always nondeterministic---every expression 
can either reduce as normal 
or halt. However, this nondeterminism 
is as benign as the multi-step judgment's nondeterminism---each
can only relate expressions to partial evaluations that arise
during typical evaluation (regardless of whether that 
typical evaluation is deterministic or not).
Moreover, this nondeterminism poses no problems for
proving with big-stop semantics, as exemplified 
by the Agda proofs referenced in \Cref{sec:eff,sec:kmach} \cite{artifact},
which each maintain the straightforward nature of their big-step counterparts.

\subsection{Properties}

The following property,
\Cref{thm:eq}, is the 
key to showing that big-stop 
semantics captures precisely the notion we are interested in.
That is, it shows big-stop semantics 
captures precisely the notion of evaluation typically built upon 
small-step semantics.

\begin{theorem}[Stop/Multi Equivalence]\label{thm:eq}
    For all expressions $e,e'$,\;\; $\bstj e {e'} \iff \msj e {e'}$
\end{theorem}

\begin{proof}
    This equivalence follows by induction on the 
    derivation of each judgment. 
\end{proof}

The key benefit of big-stop semantics is not that it 
infers any new relations beyond the multi-step judgment, but rather 
that it allows the derivation of the same multi-step relations 
in a manner more ergonomic for use in proofs.
From \Cref{thm:eq}, it immediately 
follows that any property one might state 
using multi-step semantics can be equivalently 
stated using big-stop semantics.
Already, this result covers most of the gap between 
big- and small-step semantics, and the 
only cost is a minor extension to standard big-step
operational semantics.

In light of \Cref{thm:eq}, one should also expect that
big-stopping indeed captures stopped evaluations. 
This property can be formally expressed
with the following statement, which says
that stopped expressions could continue to be evaluated 
as if they had not stopped. That is, the 
big-stop relation is transitive just like the multi-step relation.

\begin{lemma}[Transitivity]
    If $\bstj {e_1} {e_2}$ and 
    $\bstj {e_2} {e_3}$, then 
    $\bstj {e_1} {e_3}$.
\end{lemma}

\begin{proof}
This property follows from induction over the derivations.
\end{proof}

An important observation is that the big-stop 
system does not interfere with
standard big-step evaluation to values.
Call a derivation \emph{strict} if every use of the rule schema
\textsc{St-Stop(k)} in that derivation 
has the conclusion $\bstj v v$ for some value $v$.
Then there is a canonical isomorphism between strict
derivations of $\bstj e v$ and derivations of $\bsj e v$ where
each rule \textsc{St-X} corresponds to \textsc{B-X},
except for \textsc{St-Stop(0)} and \textsc{St-Stop(1)},
which correspond to \textsc{B-Val} and \textsc{B-Succ}, respectively.

\begin{lemma}[Derivation Isomorphism]\label[lemma]{lem:iso}
  For all expressions $e$ and values $v$, strict derivations of 
  $\bstj e v$ are isomorphic to derivations of $\bsj e v$.
\end{lemma}

\begin{proof}
    This property follows by induction over the structure of the derivations.
\end{proof}

\begin{corollary}[Stop/Big Equivalence] \label{cor:stbeq}
  For all expressions $e$ and values $v$,
  $\bstj e v \iff \bsj e v$.
\end{corollary}

\subsection{Progress and Preservation}

While multi-step judgments are commonly used 
incarnations of small-step semantics, there still 
exist programming language properties 
that are traditionally stated using small-step 
semantics simpliciter, without any multi-stepping.
In particular, progress and 
preservation are typically stated this way.
The rest of this section is devoted to 
recovering these statements in  
the big-stop system.

The preservation property can be restated using big-stop semantics 
as follows:

\begin{theorem}[Preservation]\label{thm:pres}
    Big-stopping preserves typing.

    That is, if $\ty \cdot e \tau$ and 
    $\bstj e {e'}$, then $\ty \cdot {e'} \tau$.
\end{theorem}

\begin{proof}
    This property can be shown by induction 
    on the derivation of the big-stop 
    judgment. 
\end{proof}

Note that preservation could also have been proven by 
first observing that the same property holds for
multi-stepping
and then applying \Cref{thm:eq}.
However, it is not necessary to make such an indirect proof;
using big-stop directly is unproblematic.

With a little additional maneuvering, one can 
also restate the property of progress.
To do so, call all of the rules ``progressing''
except for \textsc{St-Stop(k)}. 
Further, call a big-stop derivation ``progressing'' 
if it uses
a progressing rule. We use the notation 
$\bstjp e {e'}$ to mean that $\bstj {e} {e'}$ 
is the conclusion of a progressing derivation.
We can then state a big-stop version of 
the progress property as follows:

\begin{theorem}[Progress]
\label{thm:progress}
    Well-typed, non-value expressions can progress.

    That is, if $\ty \cdot e \tau$,
    then either $\val e$ or 
    $\bstjp e {e'}$ for some $e'$.
\end{theorem}

\begin{proof}
This property follows from induction over the 
structure of the expression $e$.
\end{proof}

In other words, \Cref{thm:progress} ensures that well-typed,
non-value expressions $e$ can make at least 
one small-step of
evaluation progress to some $e'$. 
To confirm that \Cref{thm:progress} really does 
capture the key notion of progress,
we also point out the following lemma:

\begin{lemma}[Progressing Means Progress]
If $\bstjp {e_1} {e_2}$, then $\exists e_3.\, \ssj {e_1} {e_3}$
and $\msj {e_3} {e_2}$.
\end{lemma}

\begin{proof}
This property follows from induction over 
the derivation of $\bstjp {e_1} {e_2}$.
\end{proof}

As a result, progressing derivations 
are the correct big-stop notion for 
making evaluation progress.
Thus, e.g., progressing derivations can be made 
arbitrarily deep iff small-step reduction loops forever.
For convenience, we use the notation 
$\bsntj{e}$ to mean that such arbitrarily 
deep progressing derivations can be made, i.e., that $e$ diverges.
This notation will help to compare to other
systems for inferring divergence in \Cref{sec:related},
but will not otherwise be used in the big-stop inference system.

\section{Effectful Big-Stopping}\label{sec:eff}

This section develops big-stop semantics with effects.
Here, we specifically consider writer-monad-style effects 
that are emitted but do not otherwise affect computation,
like printing.
Such a system can be used to solve
the effectful compilation problem presented in 
\Cref{sec:motive}. For effects that 
may interact with computation, like mutation,
see \Cref{sec:imp}.

In this section, we also prove all provided properties in Agda 
in an intrinsically typed setting~\cite{artifact}.
This formalization shows that 
 big-stop semantics makes good on its ergonomic promise:
these proofs largely follow by straightforward 
induction on the derivation of the big-stop judgment
and do not need the detours required by small-step systems.
Moreover, these proofs work in an effectful setting 
with nontermination, which previously 
was a blocker for the use of big-step semantics.

To introduce writer-monad-style effects, we extend 
PCF syntax with the expression $\eff a e$, which 
emits the effect indicated by $a$ and then continues on as $e$.
To capture these effects,
we consider sequences of them in the order 
they occur.
Thus, for example, the effect term $abc$ just means that 
effect $a$ is followed by effect $b$ and then effect $c$.
We also use $1$ to represent the identity effect, which can be 
thought
of as an uninteresting effect or no effect at all.

Such effects can be incorporated into the small-step 
semantics using the judgment $\ssej {e_1} {e_2} a$,
meaning that $e_1$ small-steps to $e_2$ and 
has the effects of $a$. The rules for this system 
are given in \Cref{fig:sses}.
The figure also defines the multi-step variant 
of this judgment,
$\msej {e_1} {e_3} {a}$.

\begin{figure}
  \small
    \begin{mathpar}
        \inferrule[SE-Seq(k)]{
            1 \leq k \leq n
            \\
            \forall\, 1\leq i < k.\, \val {e_i}
            \\
            \ssej {e_k} {e_k'} a
            \\
            \forall\, i \neq k.\, e_i = e_i'
        }{
            \ssej {E[e_1,\dots,e_n]} {E[e_1',\dots,e_n']}  a
        }

        \inferrule[SE-CaseZ]{
        }{
            \ssej {\ifz {e_1} {x} {e_2} \z} {e_1} 1
        }

        \inferrule[SE-CaseS]{
            \val v
        }{
            \ssej {\ifz {e_1} {x} {e_2} {\suc v}} {[v/x]e_2} 1
        }

        \inferrule[SE-App]{
            \val v
        }{
            \ssej {\app {\lam {f} {x} {e}} v} {[\lam {f} {x} {e}/f,\, v/x]e} 1
        }

        \inferrule[SE-Effect]{
        }{
            \ssej {\eff a e} e a
        }
        
        \inferrule[ME-Refl]{
        }{
            \msej {e} {e} 1
        }

        \inferrule[ME-Step]{
            \ssej {e_1} {e_2} a 
            \\
            \msej {e_2} {e_3} b
        }{
            \msej {e_1} {e_3} {ab}
        }
    \end{mathpar}
    \caption{Small- and multi-step semantics for effectful call-by-value PCF}
    \label{fig:sses}
    \end{figure}

The corresponding effectful big-step semantics are given by 
the rules of
\Cref{fig:bigeffect}. These rules define 
the judgment $\bsej e v a$ meaning that the expression 
$e$ evaluates to the value $v$ while emitting 
the effects indicated by $a$. 
To ensure that this system
agrees with the small-step system of 
\Cref{fig:sses}, we provide \Cref{lem:bmeeq},
the effectful analogue of \Cref{lem:bseq}.

\begin{figure}
  \small
\begin{mathpar}
  \inferrule[BE-Succ]{
    \bsej e v a
  }{
    \bsej {\suc e} {\suc v} a
  }

  \inferrule[BE-CaseZ]{ 
    \bsej e \z a
    \\ 
    \bsej {e_1} {v_1} b 
  }{ 
    \bsej  {\ifz {e_1} x {e_2} e} {v_1} {ab}
  }
    
  \inferrule[BE-CaseS]{ 
    \bsej e {\suc v} a
    \\ 
    \bsej {[v/x]e_2} {{v_2}} b
  }{ 
    \bsej  {\ifz  {e_1} x {e_2} e} {v_2} {ab}
  }
    
  \inferrule[BE-App]{ 
    \bsej {e_1} {\lam f x e} a
    \\ 
    \bsej {e_2} {v_2} b
    \\\\
    \bsej {[{\lam f x e}/f, v_2/x]e} {v} c
  }{ 
    \bsej {\app {e_1} {e_2}} {v} {abc}
  }

   \inferrule[BE-Val]{
       \val v
    }{
        \bsej {v} {v} 1
    }

  \inferrule[BE-Eff]{
    \bsej e v b
  }{
    \bsej {\eff a e} v {ab}
  }
\end{mathpar}
\caption{Big-step semantics for effectful call-by-value PCF}
\label{fig:bigeffect}
\end{figure}

\begin{lemma}[Effectful Big/Multi Equivalence\agda]\label[lemma]{lem:bmeeq}
    For all expressions $e$, values $v$, and effects $a$,
    \[\bsej e v a \iff \msej e v a\]
\end{lemma}

\begin{proof}
This property follows by induction over the structure of the derivations.
\end{proof}

\begin{figure}
  \small
\begin{mathpar}
 \inferrule[StE-Stop(k)]{
        \forall 1 \leq i \leq k. \,
        \bstej {e_{i}} {e'_{i}} {a_i}
        \\
        \forall 1 \leq i \leq k-1. \,
        \val {e'_i}
        \\\\
        \forall k+1 \leq i \leq n. \,
        e'_i = e_i
    }{
        \bstej {E[e_1, \dots, e_n]} {E[e'_1, \dots, e'_n]} {a_1\dots a_k}
    }

  \inferrule[StE-CaseZ]{ 
    \bstej e \z a
    \\ 
    \bstej {e_1} {e'_1} b
  }{ 
    \bstej  {\ifz  {e_1} x {e_2} e} {e'_1} {ab}
   }
    
  \inferrule[StE-CaseS]{ 
    \bstej e {\suc v} a
    \\ 
    \val v
    \\\\ 
    \bstej {[v/x]e_2} {{e'_2}} b
  }{ 
    \bstej  {\ifz  {e_1}  x {e_2} e} {e'_2} {ab}
  }
    
  \inferrule[StE-App]{ 
    \bstej {e_1} {\lam f x e} a
    \\ 
    \bstej {e_2} {v_2} b
    \\\\
    \val {v_2}
    \\ 
    \bstej {[{\lam f x e}/ f, v_2/x]e_2} {e'} c
  }{ 
    \bstej {\app {e_1} {e_2}} {e'} {abc}
  }

  \inferrule[StE-Eff]{
    \bstej e {e'} b
  }{
    \bstej {\eff a e} {e'} {ab}
  }
\end{mathpar}
\caption{Big-stop semantics for effectful call-by-value PCF}
\label{fig:stopeffect}
\end{figure}

With a proper big-step system in hand, 
we can now go over the big-stop extension. 
Fundamentally, the essence of big-stop semantics 
remains the same in this new setting: extend big-step 
semantics with a rule schema for stopping
(and possibly insert 
some value premisses, depending on the 
exact formulation of the big-step system).
The resulting adaptation for the big-step rules 
of \Cref{fig:bigeffect} yields the 
big-stop rules of \Cref{fig:stopeffect}.
Examples using these rules can be found 
in \Cref{sec:examples}.

\subsection{Properties}

The big-stop system can be verified to properly 
extend the given big-step system using the following lemma,
which is analogous to \Cref{lem:iso} and \Cref{cor:stbeq}.

\begin{theorem}[Effectful Stop/Step Equivalence\agda] \label{lem:stbeeq}
  For all expressions $e$, values $v$, and effects $a$,
  \[\bstej e v a \iff \bsej e v a\]
\end{theorem}

\begin{proof}
This property follows from induction over the derivations.
\end{proof}

More importantly, as described by 
\Cref{thm:stseeq}, the effectful big-stop judgment coincides with the
effectful multi-step 
judgment, just as in the pure setting's \Cref{thm:eq}.

\begin{theorem}[Effectful Stop/Multi Equivalence\agda] \label{thm:stseeq}
  For all expressions $e, e'$ and effects $a$,
  \[\bstej e {e'} a \iff \msej e {e'} a\]
\end{theorem}

\begin{proof}
This property follows from induction over the derivations.
\end{proof}

As might be expected in light of \Cref{thm:stseeq}, 
various other properties also translate cleanly into 
the effectful big-stop setting. For example, transitivity 
and progress both hold, as stated below.

\begin{lemma}[Effectful Transitivity\agda]\label[lemma]{lem:etrans}
  If $\bstej {e_1} {e_2} a$ and $\bstej {e_2} {e_3} b$,
  then $\bstej {e_1} {e_3} {ab}$
\end{lemma}
\begin{proof}
This property follows from induction over the derivations.
\end{proof}

We formalize progress with the judgment
$\bstejp e {e'} a$, which means that 
$\bstej e {e'} a$ is the conclusion of a progressing derivation
(a derivation using a rule other than \textsc{StE-Stop(k)}).

\begin{theorem}[Effectful Progress\agda]\label{thm:proge}
Well-typed, non-value expressions can progress.

    That is, if $\ty \cdot e \tau$,
    then either $\val e$ or 
    $\bstejp e {e'} a$ for some $e'$ and $a$.
\end{theorem}

\begin{proof}
  This property follows from induction over the structure of 
  the expression $e$.
\end{proof}

All of these results are proved in a direct way by 
induction on the derivation of big-stop, big-step, small-step, 
or typing judgments. The simplicity of these proofs is 
reflected in their compact formalizations in Agda \cite{artifact}, 
as shown in \Cref{fig:theorem-code-table}.

\begin{table}
  \centering
    \begin{tabular}{@{}lccccc@{}}
    \toprule
    & \Cref{lem:bmeeq}
    & \Cref{lem:stbeeq}
    & \Cref{thm:stseeq}
    & \Cref{lem:etrans}
    & \Cref{thm:proge} \\
    \midrule
    Lines of code in Agda & 64 & 28 & 88 & 30 & 18 \\
    \bottomrule
  \end{tabular}
  \caption{Size of Agda formalizations for each theorem, measured by non-comment, non-blank lines of code
  }
  \label{fig:theorem-code-table}
\end{table}

\section{Application: K Machine Correctness}\label{sec:kmach}

Now that big-stop semantics has been established, we can demonstrate
its benefits by proving semantics preservation for the K machine
compilation of effectful call-by-value PCF, as outlined in
\Cref{sec:motive}.
%
The results of this section
have been proved in Agda \cite{artifact}.

\subsection{The Effectful K Machine}

A K machine state is either of
the form $\estate{k}{e}$ or 
$\vstate k v$ where $k$ is a \textit{stack} of 
\textit{frames}, $e$ is a 
PCF expression, and $v$ is a value. 
K machine states of the form $\estate{k}{e}$
indicate that the expression $e$ is not yet a value 
and needs to be evaluated, and 
K machine states of the form $\vstate k v$
indicate that the value $v$ can be plugged into the top frame of the stack $k$.

Stacks $k$ and frames $f$ are picked out by the judgments 
$\stack k$ and $\kframe f$, respectively, described
in \Cref{fig:frames}. Each frame is just an expression 
with a hole in it indicating the subexpression 
undergoing evaluation. The rule \textsc{F-Arg} is notable because 
it has a premiss requiring the expression $v$ to 
be a value.
Strictly speaking, this premiss is not necessary,
as the K machine can only generate 
such frames where
the expression $v$ is a value. 
Thus, the value premiss only serves to exclude 
invalid stacks that could not arise anyway.
This exclusion makes the value premiss 
useful in the course of 
proving properties of the K machine, since impossible stacks 
need not be considered.

\begin{figure}
    \small
\begin{mathpar}
  \inferrule[Stack-Empty]
  { } {\stack \epsilon}

  \inferrule[Stack-Frame]{ 
    \stack k 
    \\
    \kframe f
  } { 
    \stack {k ; f}
  }
  \\
  \inferrule[F-Succ]{
  }{
    \kframe \fsucc
  }

  \inferrule[F-Case]{
  }{
    \kframe {\fcase {e_1} x {e_2}}
  }

  \inferrule[F-Fun]{
  }{
    \kframe {\ffun e}
  }

  \inferrule[F-Arg]{
    \val v
  }{
    \kframe {\farg v}
  }
\end{mathpar}
\caption{K machine stack and frames}
\label{fig:frames}
\end{figure}

The K machine's transitions are formalized 
with \Cref{fig:ke}. These rules define the judgment 
$\kesj{S}{T}{a}$ to mean that K machine state $S$ 
transitions in one step to state $T$ while emitting the
effects indicated by $a$. This figure 
also defines the multi-step 
analogue $\kemj{S}{T}{a}$.
To obtain transitions for the pure K machine,
simply ignore the effects 
and the rule \textsc{KE-Eff}.

\begin{figure}
    \small
\begin{mathpar}
  \inferrule[KE-Zero]{
  }{
    \kesj {\estate k \z}  {\vstate k \z} 1
  }

  \inferrule[KE-Succ1]{
  }{
    \kesj {\estate {k} {\suc e}} {\estate {k;\fsucc} e} 1
  }

  \inferrule[KE-Succ2]{
  }{
    \kesj {\vstate {k;\fsucc} v} {\vstate k {\suc v}} 1
  }

  \inferrule[KE-Case]{ 
  }{ 
    \kesj {\estate{k}{\ifz {e_1} x {e_2} {e_3}}} {\estate{k; \fcase {e_1} {x} {e_2} }{e_3}} 1
  }

  \inferrule[KE-CaseZ]{ 
  }{ 
    \kesj {\vstate{k; \fcase {e_1} {x} {e_2} }{\z}}  {\estate{k}{e_1}} 1
  }

  \inferrule[KE-CaseS]{ 
  }{ 
    \kesj {\vstate{k; \fcase {e_1} {x} {e_2} }{\suc v}}  {\estate{k}{[v/x]e_2}} 1
  }

   \inferrule[KE-Fun]{ 
  }{ 
    \kesj {\estate{k}{\lam f x e}} {\vstate{k}{\lam f x e}} 1
  }

  \inferrule[KE-App1]{ 
  }{ 
    \kesj {\estate k {\app {e_1} {e_2}}} {\estate {k;\ffun {e_2}} {e_1}} 1
  }

  \inferrule[KE-App2]{ 
  }{ 
    \kesj {\vstate {k;\ffun {e}} {v}} {\estate {k;\farg {v}} {e}} 1
  }

  \inferrule[KE-App3]{ 
  }{ 
    \kesj {\vstate {k;\farg {\lam f x e}} {v}}  {\estate k {[\lam f x e/f, v/x]e}} 1
  }

  \inferrule[KE-Eff]{
  }{
    \kesj {\estate k {\eff a e}} {\estate k e} a
  }

  \inferrule[KEM-Refl]{
  }{
    \kemj S S 1
  }

  \inferrule[KEM-Step]{
    \kesj S T a
    \\
    \kemj T U b 
  }{
    \kemj S U {ab}
  }
\end{mathpar}
\caption{Effectful K machine transition rules}
\label{fig:ke}
\end{figure}

\subsection{Proving Correctness}

Semantic preservation of the compilation of effectful
PCF to the K machine can be boiled down to the following four lemmas.
In particular, soundness and completeness must hold for 
both converging and diverging computations, the latter of 
which is necessary due to the presence of effects.
The rest of this subsection will briefly discuss 
their proofs, all of which have been formalized 
in Agda \cite{artifact}.

\begin{lemma}[Convergent Soundness\agda]
  \label[lemma]{lem:stsoundc}
  If $\kemj {\estate \epsilon e} {\vstate \epsilon v} a$ then $\bstej e v a$.
\end{lemma}

\begin{lemma}[Divergent Soundness\agda]
  \label[lemma]{lem:stsoundd}
  If $\kemj {\estate \epsilon e} {S} a$ then $\bstej e {e'} a$ for some $e'$
\end{lemma}

\begin{lemma}[Convergent Completeness\agda]
  \label[lemma]{lem:stcomplc}
  If $\bstej e v a$ and $\val v$ then $\kemj {\estate \epsilon e} {\vstate \epsilon v} a$.
\end{lemma}

\begin{lemma}[Divergent Completeness\agda]
  \label[lemma]{lem:stcompld}
  If $\bstej e {e'} a$ then $\kemj {\estate \epsilon e} {S} a$ for some 
  state $S$.
\end{lemma}

\Cref{lem:stsoundc,lem:stsoundd} (soundness) can both 
be proven using big-stop semantics in a 
similar manner to small-step semantics, 
as they proceed by induction over the  
K machine dynamics rather than the PCF dynamics.
Big-stop semantics is just as fit a
target for these proofs as small-step semantics.

\Cref{lem:stcomplc} (convergent completeness) is straightforward to prove 
using big-stop semantics, just as 
it was for big-step as shown in \Cref{sec:motive}.
Because big-stop 
semantics are just a small extension of 
big-step semantics, 
the thrust of the proof of \Cref{lem:stcomplc}
is practically identical 
to the easy big-step approach of \Cref{lem:compl}.

Finally, the proof of \Cref{lem:stcompld} 
(divergent completeness)
proceeds by building upon \Cref{lem:stcomplc}. 
To extend \Cref{lem:stcomplc} to cover nontermination,
the new proof just needs a new case for each of the 
rules following from the schema \textsc{StE-Stop(k)}. 
Each of these new cases is straightforward and 
proceeds much the same as a typical big-step case.
The whole proof of this lemma is therefore  
be proved directly with rule induction, 
unlike when using small- or big-step 
semantics.

Altogether, these proofs show that big-stop semantics 
is more useful than the big-step semantics it extends, and 
it also can
provide a superior proving experience as compared 
to small-step semantics.
Moreover, as these proofs concern 
the correctness of compilation, we expect the benefits 
of big-stop semantics to be widely applicable.
Finally, these proofs demonstrate 
that the nondeterministic nature of big-stop semantics 
pose no issue for formalization.

In \Cref{fig:k-machine-proof-sizes},
we compare the sizes of the Agda proofs
of the big-stop soundness and completeness lemmas
(\Cref{lem:stsoundc,lem:stsoundd,lem:stcomplc,lem:stcompld}) against 
the equivalent lemmas formulated using
big-step 
and small-step semantics. 
That is to say, we compare the sizes of proving the
soundness and completeness of the K machine with 
respect to small-step, big-step, and big-stop semantics.

\Cref{fig:k-machine-proof-sizes}'s 
results demonstrate that the big-stop proofs of 
soundness and convergent completeness are each
similar in size to 
traditional big-step proofs of 
the corresponding convergent properties.
Recall 
that big-step simply cannot handle divergent properties,
and so those table entries are left blank.
The only place big-stop 
has a larger proof, divergent completeness,
is because the divergent proof builds 
upon the convergent one, so the lines of 
the big-stop convergent completeness lemma 
are also counted toward divergent completeness.

\Cref{fig:k-machine-proof-sizes}
also shows that traditional small-step proofs of soundness 
are around the same size as their big-step and -stop 
counterparts,
which is to be expected because these proofs 
induct over K machine dynamics rather than 
PCF dynamics. Small-step completeness 
is where the onerousness of small-step semantics comes into play.
Not only are these proofs around twice as large 
as big-step/stop, but the proofs themselves 
depend on the big-step/stop proofs. 
The shortest proofs we could find
for small-step completeness involve 
translating small-step into big-step/stop 
and then invoking the completeness of big-step/stop.
Around half of the lines attributed to these proofs 
are for the translation and the other 
half are for big-step/stop
completeness. It is not clear that 
there is any better way to prove these properties 
using small-step semantics, and it is telling 
that big-stop semantics seems to be the best
way to prove divergent completeness.

\begin{table}
  \centering
  \begin{tabular}{@{}lcccc@{}}
    \toprule
    & \multicolumn{2}{c}{\textbf{Soundness}} & \multicolumn{2}{c}{\textbf{Completeness}} \\
    \cmidrule(lr){2-3} \cmidrule(lr){4-5}
    & \textbf{Convergent}
    & \textbf{Divergent} 
    & \textbf{Convergent}
    & \textbf{Divergent}  \\
    & \Cref{lem:stsoundc}
    & \Cref{lem:stsoundd}
    & \Cref{lem:stcomplc}
    & \Cref{lem:stcompld} \\
    \midrule
    \textbf{Big-stop}  & 76 & 72 & 32 & 72 \\
    \textbf{Big-step}  & 81 & -- & 29 & -- \\
    \textbf{Small-step} & 88 & 84 & 68 & 128 \\
    \bottomrule
  \end{tabular}
\caption{Comparison of Agda proof sizes, measured by non-comment, non-blank lines of code
  }
  \label{fig:k-machine-proof-sizes}
\end{table}

\subsection{The Big-Stop Method}

\Cref{lem:stcomplc,lem:stcompld} exemplify a general proof technique for
generalizing a theorem about converging 
computations to cover diverging ones too. 
We call this technique the
\emph{big-stop method}.
The starting point is a theorem that has been proved by induction on a
standard big-step judgment.
Here the appropriate theorem is the effectful version of \Cref{lem:bigcomplete}:
If $\bsej e v a$ then $\kesj {\estate k e} {\vstate k v} a$.  

The methodical generalization to diverging computations proceeds in three
steps.
First, extend the big-step judgment to big-stop
by adding the appropriate stopping rules.
Second, reproduce the proof of the original 
theorem using the new big-stop
judgment, yielding a proof concerning values (\Cref{lem:stcomplc}).
Third, add induction cases for the new stopping rules, building 
upon the previous proof (\Cref{lem:stcompld}).
All together, 
these steps result in the generalized theorem for partial
evaluations.

The key---which demonstrates the advantage of big-stop over
small-step---is that the generalized 
theorem is void of complicated 
inductive invariants.
In \Cref{lem:stcompld}, the example at hand, this 
materializes in the
relationship (or lack thereof) between 
the partially-evaluated expression $e'$ and state $S$.
In a multi-step setting, one would need
to relate $e'$ and $S$ because they would
eventually occur on the left-hand side 
of relations where they would be inductively consumed.
An invariant would then be necessary to ensure 
that the correct properties are maintained through this consumption.
In contrast, the big-stop system never consumes any
expression it puts out unless that expression is a value,
which is much simpler to deal with than 
an arbitrary partially-evaluated expression---in 
fact, it can be dispatched with the theorem 
for values (here \Cref{lem:stcomplc}), with no complicated 
invariant needed.
Thus, the proof of 
the generalized theorem can proceed by straightforward
induction.


\section{Further Ergonomic Optimization}\label{sec:opt}

This section presents some additional variants 
of big-stop systems that further minimize the 
extent of the change from a typical big-step system.
As set up previously in the paper, big-stop semantics already
essentially requires only one new rule schema, but 
this schema is most easily dispatched in practice 
by replacing it with a rule for each of its 
instantiations---we take this approach with our Agda proofs.
This replacement results in a number of additional rules proportional 
to the number of language constructs, and each of 
these rules needs its own proof case.
This ballooning of proof obligations when using big-stop semantics
can be tedious, even if each new case is easy to dispatch.

With fewer changes, the proof obligations of such a system 
are closer to those of the original big-step system. 
These obligations can be optimized to a point where 
one practically gets the nonterminating results of big-stop 
semantics for free after establishing the standard 
big-step result. At the limit, big-stop proofs
only require extending big-step proofs with 
one easy proof case.

\subsection{Normal Form}

One way to minimize the number of rules required in big-stop 
semantics is to only consider languages that are in
monadic normal form~\cite{jones1993partial, hatcliff1994generic}
(a close relative of A-normal form~\cite{sabry1992reasoning,flanagan1993essence}).
Monadic normal form requires that all possible subexpressions 
are either values or let-bound variables. 
This normalization does not affect expressivity of the language.
An example of monadic normal form PCF expressions is given by the 
grammar of $e$ in \Cref{fig:monadiclang}.

\begin{figure}
  \small
    \begin{subfigure}{0.4\textwidth}
    \begin{align*}
    v ::= & \; x \mid \lam {f} {x} e \mid \z \mid \suc v 
    \end{align*}
    \caption{Variables and values}
    \end{subfigure}
    \begin{subfigure}{0.4\textwidth}
    \begin{align*}
    e ::= & \;v \mid \ifz {e_1} x {e_2} v \mid \app {v_1} {v_2} \mid \elet x {e_1} {e_2} 
    \end{align*}
    \caption{Expressions}
    \end{subfigure}
\caption{Language of monadic normal PCF}
\label{fig:monadiclang}
\end{figure}

Monadic normal form simplifies the dynamics of 
computation by removing most
congruences. Instead, monadic normal form
uses let-expressions to explicitly sequence computations
in a syntactic way. This approach
means the only congruence rule needed is that for
the let-expression itself.
See the small-step
semantics of \Cref{fig:monadicsmall} where the only 
congruence rule is \textsc{SM-Let1}.

\begin{figure}
  \small
    \begin{mathpar}
        \inferrule[SM-Let1]{
            \ssj {e_1} {e_1'} 
        }{
            \ssj {\elet x {e_1} {e_2}} {\elet x {e_1'} {e_2}} 
        }

        \inferrule[SM-Let2]{
            \val v
        }{
            \ssj {\elet x v e} {[v/x]e} 
        }

        \inferrule[SM-CaseZ]{
        }{
            \ssj {\ifz {e_1} {x} {e_2} \z} {e_1} 
        }

        \inferrule[SM-CaseS]{
        }{
            \ssj {\ifz  {e_1} {x} {e_2} {\suc v}} {[v/x]e_2} 
        }

        \inferrule[SM-App]{
        }{
            \ssj {\lam {f} {x} {e}\, v} {[\lam {f} {x} {e}/f,\, v/x]e} 
        }
    \end{mathpar}
    \caption{Small-step semantics for monadic normal PCF}
    \label{fig:monadicsmall}
    \end{figure}

Both the small-step semantics of \Cref{fig:monadicsmall} 
and the corresponding big-step semantics \Cref{fig:monadicbig}
also exhibit another simplification of monadic normal form:
fewer premisses are required. Because monadic normal 
form forces subexpressions to be values as much as possible,
no premisses need to be spent on their evaluation.
Compare \Cref{fig:monadicsmall,fig:monadicbig} to 
their counterparts \Cref{fig:sss,fig:bigpure}.

\begin{figure}
  \small
    \begin{mathpar}
        \inferrule[BM-Val]{
       \val v
    }{
        \bsj {v} {v} 
    }
    
        \inferrule[BM-Let]{
            \bsj {e_1} {v_1} 
            \\
            \bsj {[v_1/x]e_2} {v_2}
        }{
            \bsj {\elet x {e_1} {e_2}} {v_2} 
        }
  \\
  \inferrule[BM-CaseZ]{ 
    \bsj {e_1} {v_1} 
  }{ 
    \bsj  {\ifz {e_1} x {e_2} \z} {v_1} 
  }
    
  \inferrule[BM-CaseS]{ 
    \bsj {[v/x]e_2} {{v_2}} 
  }{ 
    \bsj  {\ifz  {e_1} x {e_2} {\suc v}} {v_2}
  }
    
  \inferrule[BM-App]{ 
    \bsj {[{\lam f x e}/f, v/x]e} {v'} 
  }
  { \bsj {\app {\lam f x e} {v}} {v'}}
    \end{mathpar}
    \caption{Big-step semantics for monadic normal PCF}
    \label{fig:monadicbig}
    \end{figure}

The appropriate big-stop extension of 
\Cref{fig:monadicbig} is given by \Cref{fig:monadicstop},
which only contains one more rule in total.
This extension allows stopping at any point with the two
rules \textsc{StM-Stop} and \textsc{StM-Let1}.
The rule \textsc{StM-Stop} subsumes \textsc{BM-Val} much like 
\textsc{St-Stop(0)} subsumes \textsc{B-Val}, and the rule 
\textsc{StM-Let1} is the single congruence 
rule needed to propagate stopping. 
The rule \textsc{StM-Let2} adds a value premiss 
to the rule \textsc{BM-Let}, just as was done to many rules 
in previous big-stop extensions. This time, however,
\textit{only} the rule \textsc{StM-Let2} needs a new value 
premiss, and the rest of the rules can remain totally untouched.
The rules \textsc{StM-Let1} and \textsc{StM-Let2} together 
are respectively
analogous to the rules \textsc{SM-Let1} and \textsc{SM-Let2}
of \Cref{fig:monadicsmall}, much like 
how \textsc{St-Stop(k)} is analogous to \textsc{S-Seq(k)}.

\begin{figure}
  \small
    \begin{mathpar}
        \inferrule[StM-Stop]{
    }{
        \bstj {e} {e} 
    }

     \inferrule[StM-Let1]{
            \bstj {e_1} {e_1'} 
        }{
            \bstj {\elet x {e_1} {e_2}} {\elet x {e_1'} {e_2}}
        }

        \inferrule[StM-Let2]{
            \bstj {e_1} {v_1} 
            \\
            \val {v_1}
            \\
            \bstj {[v_1/x]e_2} {e_2'}
        }{
            \bstj {\elet x {e_1} {e_2}} {e_2'} 
        }

  \inferrule[StM-CaseZ]{ 
    \bstj {e_1} {e_1'} 
  }{ 
    \bstj  {\ifz {e_1} x {e_2} \z} {e_1'} 
  }
    
  \inferrule[StM-CaseS]{ 
    \bstj {[v/x]e_2} {{e_2'}} 
  }{ 
    \bstj  {\ifz  {e_1} x {e_2} {\suc v}} {e_2'}
  }
    
  \inferrule[StM-App]{ 
    \bstj {[{\lam f x e}/f, v/x]e} {e'} 
  }{ 
    \bstj {\app {\lam f x e}  {v}} {e'}
  }
    \end{mathpar}
    \caption{Big-stop semantics for monadic normal PCF}
    \label{fig:monadicstop}
\end{figure}

The key takeaway of using monadic normal form
is that the unwieldy rule schema for sequencing can be replaced with 
rules for let-bindings rather than rules for every possible 
congruence. The resulting system overall 
has just one more rule in total than 
its big-step base and lightly touches two existing rules
(\textsc{BM-Val} is loosened to \textsc{StM-Stop} and 
\textsc{BM-Let} gains an explicit value premiss as 
\textsc{StM-Let2}). 
Therefore, rather than dispatch a rule schema like \textsc{St-Stop(k)} 
by roughly doubling the number of rules to consider in proof cases, 
monadic normal form enables one to only consider a fixed number of 
new rules.
This accounting applies 
to any language that can be put into the appropriate 
monadic normal form,
not just PCF.

\subsection{Evaluation Contexts}

A similar effect to using monadic normal 
form can be induced via adapting
big-step 
semantics to evaluation contexts. 
Such an approach reduces the big-stop 
extension to just one rule replacement.

The appropriate meta-syntax for PCF's
evaluation contexts is given by $C$ in 
the following grammar, where 
$\hole$ is a hole and where $e$ and $v$ are 
expressions and values as defined in 
\Cref{fig:lang}. This meta-syntax allows us to write $C\langle e \rangle$ to
represent the expression that results from 
replacing the hole in $C$ with $e$, so that,
e.g., $\suc \hole \langle \z \rangle = \suc \z$.
\[
    C \;::=\; \hole  \mid \suc C \mid \ifz e x e C \mid \app C e \mid \app v C
\]

The rules of \Cref{fig:ecstop}
save for \textsc{EC-Stop} then yield
an evaluation context system equivalent to big-step semantics.
Replacing \textsc{EC-Val} with 
\textsc{EC-Stop} yields the corresponding
big-stop semantics.

In this formulation, the rules \textsc{EC-CaseZ}, \textsc{EC-CaseS},
and \textsc{EC-App} are almost the same 
as their monadic normal counterparts. The only difference is
that certain subexpressions are no longer 
syntactically guaranteed to be values 
and so require value premisses. These rules 
therefore function in the same way as their counterparts,
describing how to compute redexes.

The workhorse of the evaluation context formulation 
is the rule \textsc{EC-Seq}. This rule handles all congruences 
by finding the next subexpression ready to 
take one step,
which is picked out with a hole. 
However, big-step semantics needs to consider 
the result of many steps of evaluation, not just one,
and subsequent steps of evaluation may not be local 
to the first. The rule \textsc{EC-Seq} handles this 
by chaining together sequences of 
evaluation steps, much like the rule \textsc{M-Step} 
does for multi-step semantics and let-bindings 
do for monadic normal form. 
First the hole 
contents $e_1$ is evaluated, resulting 
in $e_1'$. Then the resulting
expression $C\langle e_1' \rangle$ is reassessed for evaluation, allowing 
the meta-syntax to readjust the hole to the next part 
of the expression to evaluate.

\begin{figure}
  \small
\begin{mathpar}
    \inferrule[EC-Stop]{
    }{
        \bstj e e
    }

    \inferrule[EC-Val]{
        \val v
    }{
       \bstj v v
    }

    \inferrule[EC-Seq]{
        \bstj {e_1} {e_1'}
        \\
        \bstj {C\langle e_1'\rangle} {e_2}
    }{
        \bstj {C\langle e_1 \rangle} {e_2}
    }
  \\ 
    \inferrule[EC-CaseZ]{
        \bstj {e_1} {e_1'}
    }{
        \bstj {\ifz {e_1} x {e_2} \z} {e_1'}
    }

    \inferrule[EC-CaseS]{
        \val {v}
        \\
        \bstj {[v/x]e_2} {e_2'}
    }{
        \bstj {\ifz {e_1} x {e_2} {\suc {v}}} {e_2'}
    }

    \inferrule[EC-App]{
        \val {v_2}
        \\
        \bstj {[\lam f x {e}/f,\, v/x]e} {e'}
    }{
        \bstj {\app {\lam f x {e}}{v}} {e'}
    }
\end{mathpar}
\caption{
    Big-stop semantics for evaluation context PCF
}
\label{fig:ecstop}
\end{figure}

This approach does 
have at least one theoretical drawback resulting from 
how the system contorts around the rule \textsc{EC-Seq}.
This rule does not guarantee 
evaluation progress, as \textsc{EC-Stop} or \textsc{EC-Val} 
could satisfy its left-hand premiss, leaving the righthand 
premiss equal to the conclusion. 
As a result, arbitrarily deep non-progressing derivations exist 
for all well-typed expressions, which is not 
a desirable property. Nonetheless, if one 
is willing to introduce some 
additional complexity, this issue could 
be alleviated by mutually inductively defining 
the big-stop judgment alongside the progressing 
judgment.

\subsection{Annihilator Effect}

It is often sufficient to only consider the effects emitted by
computation rather than any value resulting from computation.
(Even if one wishes to consider resulting values, 
information about such values can be embedded into 
the emitted effects.)
In such an effectful setting, one can leverage 
the algebraic structure of effect sequences 
to simplify away most
stopping rules.

To make sense of the algebraic approach
provided in this section, it is necessary to 
consider the structure of the effects 
represented by the term $a$. Such a term represents a 
chronological sequence of effects.
Such a sequence can be treated as a \textit{monoid},
which is an algebraic structure with an 
associative binary operation and an identity element.
In this case, the binary operation is 
sequencing---$ab$ means the effects of $b$ follow those of $a$---and 
$1$ represents the identity effect.

To support big-stop semantics,
the monoid of effects needs to be extended 
with one new element not used in the 
original small-step system: $0$. This new element
does not need to have any actual incarnation as 
an effect and can instead be thought of as 
an algebraic bookkeeping device. The new $0$ element
acts 
as a left annihilator for the monoid,
so that $0a = 0$ for any $a$.
The fact that it is a \textit{left} annihilator
rather than an annihilator simpliciter 
means that $a0$ does not 
necessarily equal $0$.
As a result, a string of effects including $0$s
appears to cut the string after the first $0$:
$abc0def = abc0$.

\begin{figure}
  \small
\begin{mathpar}
  \inferrule[StA-Succ]{
    \bstej e {e'} a
  }{
    \bstej {\suc e} {\suc {e'}} a 
  }

  \inferrule[StA-CaseZ]{ 
    \bstej e \z a
    \\ 
    \bstej {e_1} {v_1} b 
  }{ 
    \bstej  {\ifz {e_1} x {e_2} e} {v_1} {ab}
  }
    
  \inferrule[StA-CaseS]{ 
    \bstej e {\suc v} a
    \\ 
    \bstej {[v/x]e_2} {{v_2}} b
  }{ 
    \bstej  {\ifz  {e_1} x {e_2} e} {v_2} {ab}
  }
    
  \inferrule[StA-App]{ 
    \bstej {e_1} {\lam f x e} a
    \\ 
    \bstej {e_2} {v_2} b
    \\\\
    \bstej {[{\lam f x e}/f, v_2/x]e} {v} c
  }{ 
    \bstej {\app {e_1} {e_2}} {v} {abc}
  }

  \inferrule[StA-Eff]{
    \bstej e v b
  }{
    \bstej {\eff a e} v {ab}
  }

      \inferrule[StA-Val]{
       \val v
    }{
        \bstej {v} {v} 1
    }

       \inferrule[StA-Stop]{
       \val v
    }{
        \bstej {e} {v} 0
    }
\end{mathpar}
\caption{Big-stop semantics with the annihilator}
\label{fig:newe}
\end{figure}

To then extend effectful big-step rules into big-stop rules,
only one new rule needs to be added, and no other rules 
need to be touched. This new rule should say that 
an expression can always 
halt evaluation at an arbitrary value by emitting the 0 effect.
For example, extending the big-step rules of
\Cref{fig:bigeffect} to big-stop
yields \Cref{fig:newe}, where the single new stopping rule
is \textsc{StA-Stop}. Every other rule is identical 
to its big-step counterpart.

This treatment 
allows the existing big-step rules to naturally
propagate stoppage since 
every halted expression is just some value. This 
compatibility is the reason
that this big-stop
system can get by with only one stopping
rule---the resulting arbitrary value is
treated just as any other value 
and the existing rules
can compose the effects of its (partial) 
evaluation with no added difficulty.

One might also be concerned that this approach 
could be unsafe because any value at all can be chosen
for the expression $e$ in \textsc{StA-Stop}, so computation might not 
continue as the typical semantics would dictate. However,
because stopping emits the 
$0$ effect, any computation 
that follows stopping has its effects
annihilated. Moreover, one can use the presence of 
the 0 effect to determine whether the value 
resulting from big-stopping is spurious or not.
The effects include 0 iff the stopping rule is used 
in the derivation, since 0 is a new effect 
not present in the small- or big-step semantics. Thus 
if there is a 0, the resulting value is meaningless,
and if there is no 0, the resulting value matches 
that found by the big-step semantics.
One might
achieve the same result by replacing the arbitrary 
value with a dummy token, but then the dynamics 
would need to account for premisses that 
expect certain forms of values, such as 
the premiss $\bstej e \z a$ of \textsc{StA-CaseZ}.

With the new stopping rule in place, 
the following correspondence can be drawn between 
effectful big-stop semantics and 
small-step semantics:

\begin{theorem}[Annihilating Stop/Multi Equivalence]
    \label{thm:annihilatoreq}
For all expressions $e$ and effects $a$,
\[ 
    (\exists {e_1}.\, \bstej{e} {e_1} {a0})
    \iff
    (\exists {e_2}.\, \msej{e} {e_2} a)
\]
\end{theorem}

\begin{proof}
This equivalence follows by induction on 
the derivation of each judgment. 
\end{proof}

Unlike the previous effectful equivalence provided,
\Cref{thm:stseeq},
\Cref{thm:annihilatoreq} does not express 
an exact equivalence between big-stopping 
and multi-stepping. 
In particular, the expressions $e_1$
and $e_2$ bear no clear relation, and 
the effect trace induced in the big-stop 
judgment ends in 0. (The multi-step 
judgment's effect cannot include any 0s because 
0 is a new effect introduced only via 
\textsc{StA-Stop}.) Nonetheless,
\Cref{thm:annihilatoreq} shows that 
the key focus of this setting, the effects of $a$,
are clearly maintained through each semantic system. 
As a result, \Cref{thm:annihilatoreq} can 
be used similarly to previous equivalences to 
recover most of the value of small-step 
semantics in this setting.



\section{Imperative Variant}\label{sec:imp}

To allay any concerns that the big-stop approach 
only applies in nicely-behaved functional 
settings with impoverished effects,
this section shows how to apply big-stop techniques
to an imperative while-loop language with mutation.

\paragraph{An Imperative Language}

The language we consider has the following statements, where $x$
ranges over variable names and $a$ stands for integer-valued
arithmetic expressions that contain variables (e.g., $x + 5$).
For simplicity, all variables stand for integers.
\[ s ::= \iif a s \mid \iset x a \mid \iwhile a s \mid \iseq s s \mid \iskip \]
%
%
The guards of the control-flow statements, namely if-then conditionals
and while-do loops, test for the inequality of a given arithmetic
expression with zero.
Compound statements are sequenced using
$\iseq {s_1} {s_2}$.
The language is also equipped with the identity 
statement $\iskip$, which does nothing.

\begin{figure}
\small
\begin{mathpar}
\inferrule[SI-Bind]{
    \sigma(a) = z
}{
    \ssj {\langle  {\iset x a} \mid \sigma \rangle} {\langle \iskip \mid [z/x]\sigma  \rangle}
}

\inferrule[SI-Seq1]{
    \ssj {\langle s_1 \mid \sigma \rangle} {\langle s_1' \mid \sigma' \rangle}
}{
    \ssj {\langle \iseq {s_1} {s_2} \mid \sigma \rangle} {\langle \iseq {s_1'} {s_2} \mid \sigma'\rangle}
}

\inferrule[SI-Seq2]{
}{
    \ssj {\langle \iseq {\iskip} {s} \mid \sigma \rangle} {\langle s \mid \sigma\rangle}
}

\inferrule[SI-Then]{
    \sigma(a) \neq 0
}{
    \ssj {\langle  {\iif a s} \mid \sigma \rangle} {\langle  s \mid \sigma\rangle}
}

\inferrule[SI-Else]{
    \sigma(a) = 0
}{
    \ssj {\langle {\iif a s} \mid \sigma \rangle} {\langle {\iskip} \mid \sigma\rangle}
}

\inferrule[SI-Do]{
    \sigma(a) \neq 0
}{
    \ssj {\langle {\iwhile a s} \mid \sigma \rangle} {\langle  {\iseq s {\iwhile a s}}\mid \sigma\rangle}
}

\inferrule[SI-Done]{
    \sigma(a) = 0
}{
    \ssj {\langle {\iwhile a s} \mid \sigma \rangle} {\langle \iskip \mid \sigma\rangle}
}
\end{mathpar}
\caption{Small-step semantics for the imperative language}
\label{fig:smallimp}
\end{figure}

The small-step semantics of this imperative language is given 
by \Cref{fig:smallimp}. In this setting, a step relates pairs
of statements $s$ and states $\sigma$.
A state is a mapping of variables to integers. We use the notation 
$\sigma(a)$ to evaluate the arithmetic expression $a$ under those mappings,
and we use the notation $[z/x]\sigma$ to update $\sigma$ to include 
a binding of $z$ to $x$. 
Note that, upon termination, one is always left 
with the statement $\iskip$; the
computational result is recorded via the paired state.

The big-step semantics defines the judgment
$\bsj {\langle s \mid \sigma \rangle} {\sigma'}$,
which states that statement $s$ terminates in state $\sigma'$ when
executed with starting state $\sigma$.
The syntax directed rules for the judgment are identical to the rules
in \Cref{fig:bigimp} when we omit the first component of the evaluation
result.

\begin{figure}
    \small
\begin{mathpar}
\inferrule[BI-Skip]{
}{
    \bstj {\langle \iskip \mid \sigma \rangle} {\langle \iskip \mid \sigma \rangle}
}

\inferrule[BI-Then]{
    \sigma(a) \neq 0
    \\
    \bstj {\langle {s} \mid {\sigma} \rangle} {\langle s' \mid \sigma' \rangle}
}{
    \bstj {\langle {\iif a {s}} \mid {\sigma} \rangle} {\langle s' \mid \sigma' \rangle}
}

\inferrule[BI-Else]{
    \sigma_1(a) = 0
}{
    \bstj {\langle {\iif a {s}}  \mid {\sigma} \rangle} {\langle {\iskip} \mid \sigma \rangle}
}

\inferrule[BI-Seq]{
    \bstj  {\langle s_1  \mid {\sigma_1} \rangle} {\langle \iskip  \mid {\sigma_2} \rangle} 
    \\
    \bstj  {\langle s_2  \mid {\sigma_2} \rangle} {\langle s'  \mid {\sigma_3} \rangle} 
}{
    \bstj {\langle {\iseq {s_1} {s_2}}  \mid {\sigma_1} \rangle} {\langle s' \mid \sigma_3 \rangle}
}

\inferrule[BI-Bind]{
    \sigma(a) = z
}{
    \bstj {\langle {\iset x a} \mid \sigma \rangle}  {\langle {\iskip} \mid [z/x]\sigma \rangle}
}

\inferrule*[leftskip=2em,rightskip=2em,lab=BI-Do]{
    \sigma_1(a) \neq 0
    \\
    \!\!\!\! \bstj {\langle {s} \mid {\sigma_1} \rangle} {\langle \iskip \mid \sigma_2 \rangle} \!\!\!\!
    \\
    \bstj {\langle  \iwhile a {s} \mid {\sigma_2} \rangle} {\langle s' \mid \sigma_3 \rangle}
}{
    \bstj {\langle  \iwhile a {s} \mid {\sigma_1} \rangle} {\langle s' \mid \sigma_3 \rangle}
}

\inferrule*[rightskip=2em,lab=BI-Done]{
    \sigma_1(a) = 0
}{
    \bstj {\langle {\iwhile a {s}} \mid {\sigma} \rangle}{\langle \iskip \mid \sigma \rangle}
  }
  
\end{mathpar}
\caption{Big-stop versions of the big-step rules of the imperative language}
\label{fig:bigimp}
\end{figure}

\paragraph{Big-Stop Semantics}

Following the same principle as for PCF, we extend the big-step rules
for the imperative language in two steps.
To match the format of the small-step semantics, we change the results
of computations to be pairs $\langle s \mid \sigma \rangle$ of
statements $s$ and states $\sigma$ instead of simply states like in
standard big-step semantics.
The big-stop versions of the big-step rules are given in
\Cref{fig:bigimp}.
Because big-step rules only cover terminating computations,
these rules derive judgments of the form
$\bstj {\langle s \mid \sigma \rangle} {\langle \iskip \mid \sigma'
  \rangle}$ where the resulting statement is always $\iskip$.
To extend the judgment to partial evaluations, one then needs to add
the stopping rules of \Cref{fig:stopimp}. As usual,
there is one rule for stopping any program where it stands
(\textsc{StI-Stop}) and a few additional rules for propagating that
stoppage into the middle of the program. Note that \textsc{StI-Stop}
can also just replace \textsc{BI-Skip}.

\begin{figure}
    \small
\begin{mathpar}
\inferrule[StI-Seq]{
    \bstj {\langle s_1 \mid {\sigma} \rangle} {\langle {s_1'} \mid \sigma' \rangle}
}{
    \bstj {\langle \iseq {s_1} {s_2} \mid {\sigma} \rangle} {\langle {\iseq {s_1'} {s_2}} \mid \sigma' \rangle}
}

\inferrule[StI-Do]{
    \sigma(a) \neq 0
    \\
    \bstj {\langle {s} \mid {\sigma} \rangle} {\langle {s'} \mid \sigma' \rangle}
}{
    \bstj {\langle  {\iwhile a {s}} \mid {\sigma} \rangle} {\langle {\iseq {s'} {\iwhile a {s}}} \mid \sigma' \rangle}
}

\inferrule[StI-Stop]{
}{
    \bstj {\langle s \mid \sigma \rangle} {\langle s \mid \sigma \rangle}
}
\end{mathpar}
\caption{New rules for big-stop semantics for the imperative language}
\label{fig:stopimp}
\end{figure}

As expected, these big-stop semantics agree with 
the small-step semantics (\Cref{thm:impeq}).

\begin{theorem}[Imperative Stop/Multi Equivalence]\label{thm:impeq}
    For all statements $s,s'$ and states $\sigma, \sigma'$,
    \[\bstj {\langle s \mid \sigma \rangle} {\langle s' \mid \sigma' \rangle}
    \iff \msj {\langle s \mid \sigma \rangle} {\langle s' \mid \sigma' \rangle}\]
\end{theorem}

\begin{proof}
This property follows from induction over the derivations.
\end{proof}

In many formalisms, imperative big-step judgments
are of the form 
$\bsj {\langle s \mid \sigma \rangle} {\sigma'}$
where $s$ is a statement and $\sigma,\sigma'$ are states.
This form of judgment can also be extended 
to big-stop, in particular by using the annihilator optimization
from \Cref{sec:opt}. 
We exemplify this extension in
\Cref{sec:extraimp}.


\section{Related Work}\label{sec:related}

Big-stop semantics is not the first system 
designed to ameliorate the shortcomings of big-step 
semantics. Various extensions to 
big-step semantics already
exist in the literature. Broadly speaking,
these extensions operate along two axes:
interpreting big-step rules coinductively,
and adding additional rules for handling 
the problematic cases. We describe a variety 
of these systems in this section and display 
what their rules look like using variants 
of the ``big-step'' 
judgment $e \Downarrow \nontinf$,
meaning that the expression $e$ does not terminate.

\paragraph{Coinductive Techniques}
Cousot and Cousot were the first 
to propose using coinductive 
techniques to handle infinitary 
computation in the context of big-step semantics \cite{cousot1992inductive}. 
They did so by introducing additional, 
coinductively-interpreted big-step
rules to capture nontermination.
Put simply, their rules
define nontermination of 
some expression $e$ via the nontermination 
of some subevaluation of $e$.  
In principle, this approach appears to require
a new set of such rules
for each possibly-nonterminating 
subexpression of a given syntactic form,
which often more than doubles the 
number of inference rules. An example of 
their style of nontermination rules
is given
in \Cref{fig:cousotrules} for 
application.\footnote{
    Recall that $\bsntj e$ 
    means that the reduction of $e$
    does not terminate.
}
There are three such rules that are added
alongside the single terminating
big-step application rule.
These coinductive rules form the basis 
of several proposals to 
characterize nonterminating 
computation, including
\citet{hughes1995making,crole1998lectures},
and others.

\begin{figure}
    \small
\begin{mathpar}
\inferrule[CC-App1]{
    \bsntj {e_1}
}{
    \bsntj {\app {e_1}{e_2}}
}

\inferrule[CC-App2]{
    \bsj {e_1} {v_1}    
    \\
    \bsntj {e_2} 
}{
    \bsntj {\app {e_1}{e_2}} 
}

\inferrule[CC-App3]{
    \bsj {e_1} {\lam {f} {x} e}    
    \\
    \bsj {e_2} {v}
    \\\\
    \bsntj {[(\lam {f} {x} e)/f, v/x]e} 
}{
    \bsntj {\app {e_1}{e_2}} 
}
\end{mathpar}
\caption{Cousot-Cousot-style \cite{cousot1992inductive} coinductive divergence rules for application}
\label{fig:cousotrules}
\end{figure}

Leroy and Grall \cite{leroy2009coinductive} formalize
and compare both
Cousot and Cousot's \cite{cousot1992inductive}
coinductive treatment of nontermination 
and a notion of
``coevaluation'' arising from 
coinductive interpretation
of the standard eager big-step rules.
This coevaluation is able to 
handle divergence
coinductively in a
similar manner to the rules 
of \Cref{fig:cousotrules},
but simultaneously 
it can reason about 
the values that result from
terminating evaluation.
Some key relations found
between evaluation, coevaluation, and 
divergence include that
if $e$ evaluates to $v$ then 
$e$ also coevaluates to $v$,
and that if $e$ coevaluates to $v$
then $e$ either evaluates to $v$ or 
diverges. However, there exist 
expressions that diverge but do 
not coevaluate, and 
big-step coevaluation 
does not perfectly coincide with 
small-step coevaluation.

Later work by Cousot and Cousot 
approaches the problem 
using fixedpoints in a different 
way that is neither inductive nor 
coinductive \cite{cousot2009bi}.
The resulting ``bifinitary'' 
system is able to handle 
both terminating and nonterminating 
big-step-style inference rules 
using just one kind of judgment,
rather than separate terminating 
and nonterminating ones. 
As a result, this system requires less 
rule duplication than their previous 
system, but still requires multiple 
rules per syntactic form. 
Cousot and Cousot provide a three-rule setup 
for reasoning about 
function application in section 6.4.4
of their work \cite{cousot2009bi}.
The three rules are largely the same as 
those of
\Cref{fig:cousotrules} except 
that their version of \textsc{CC-App3} is 
generalized to allow 
for both terminating and nonterminating 
executions in its third premiss and 
conclusion, so long as the result of each reduction is the same.

Chargu\'eraud's pretty-big-step semantics 
\cite{chargueraud2013pretty}
takes a similar approach to 
Cousot and Cousot's bifinitary system 
\cite{cousot2009bi} but ensures 
that the interpretation of its 
rules correspond 
to standard induction and coinduction. 
In particular, interpreting pretty-big-step 
rules inductively corresponds to 
evaluation (which handles termination), 
and interpreting the same rules 
coinductively corresponds to 
coevaluation (which handles nontermination).
The pretty-big-step system also builds upon 
another technique in Cousot and Cousot \cite{cousot2009bi}:
breaking down expression reduction 
into smaller chunks. In Chargu\'eraud's work, 
this breakdown minimizes
premiss duplication across rules. 
Chargu\'eraud then goes further to 
capture features like
divergence and exceptions 
alongside values via ``outcomes.''
The new outcome for, e.g., divergence 
is an additional pseudo-value 
$\nontinf$ which requires its own propagation rules.
An example of a resulting rule set can be found 
in \Cref{fig:prettybigrules},
where application is broken down into 
$\mathtt{app}$, $\mathtt{app2}$,
and $\mathtt{app3}$, and where 
$o$ ranges over outcomes.
While 
this approach results in more
inference rules,  
Chargu\'eraud reports that having fewer premisses leads 
to smaller formal definitions and proofs.
Chargu\'eraud also proves similar theorems 
about pretty-big-step coevaluation as Leroy 
and Grall prove about their coevaluation \cite{leroy2009coinductive}.
However, because pretty-big-step semantics 
can coevaluate to the nonterminating outcome $\nontinf$,
all diverging expressions can be shown 
to coevaluate.

\begin{figure}
    \small
\begin{mathpar}
\inferrule[Ch-App1]{
    \bsj {{e_1}} {o_1}
    \\
    \bsj {\mathtt{app2}\;{o_1}\;{e_2}} v
}{
    \bsj {\mathtt{app}\;{e_1}\;{e_2}} v
}

\inferrule[Ch-App2]{
    \val {v_1}
    \\
    \bsj {{e_2}} {o_2}
    \\
    \bsj {\mathtt{app3}\;{v_1}\;{o_2}} v
}{
    \bsj {\mathtt{app2}\;{v_1}\;{e_2}} v
}

\inferrule[Ch-App2Div]{
}{
    \bsntj {\mathtt{app2}\;{\nontinf}\;{e_2}} 
}

\inferrule[Ch-App3]{
    \bsj {[\lam {f} {x} e/f, v/x]e} {o}
}{
    \bsj {\mathtt{app3}\;{\lam {f} {x} e}\;{v}} {o}
}

\inferrule[Ch-App3Div]{
}{
    \bsntj {\mathtt{app3}\;{\lam {f} {x} e}\;{\nontinf}} 
}
\end{mathpar}
\caption{Pretty-big-step-style \cite{chargueraud2013pretty}
inductive/coinductive rules for application with divergence 
 }
\label{fig:prettybigrules}
\end{figure}

Similar to 
pretty-big-step semantics is
the flag-based semantics 
of Poulsen and Mosses \cite{poulsen2017flag}.
They modify pretty-big-step 
semantics by pairing expressions with 
status flags that
propagate information 
about termination:
$\downarrow$ for successful termination 
and $\uparrow$ for divergence.
This approach reduces the number 
of propagation rules needed 
because, rather than allow 
a range of outcomes, the result of 
evaluation is always some value that 
can be propagated naturally by the existing 
rules. It might just be the case that 
the resulting value is ignored 
because evaluation 
is flagged as diverging.
An example of flag-based 
rules applied to application can 
be found in \Cref{fig:flagrules},
where $\delta$ ranges 
over flags.
There are three rules for 
application, and the rules 
\textsc{PM-Div} and \textsc{PM-Val} 
are also included 
to show how flags are handled.
Note that a coinductive interpretation
of these rules 
is necessary to derive
nontermination judgments. Such 
nontermination judgments take the form
$\bsj {(e, \downarrow)}{(v, \uparrow)}$
and ignore the value $v$ 
so that they have the same meaning 
as $\bsntj e$.

\begin{figure}
    \small
    \begin{mathpar}
    \inferrule[PM-Div]{
        \val v
    }{
        \bsj {(e, \uparrow)} {(v, \uparrow)}
    }

    \inferrule[PM-Val]{
        \val v
    }{
        \bsj {(v, \downarrow)} {(v, \downarrow)}
    }

    \inferrule[PM-App1]{
        \bsj {({e_1}, \downarrow)} {(v_1, \delta)}
        \\
        \bsj {(\mathtt{app2}\;{v_1}\;{e_2}, \delta)} {(v, \delta')}
    }{
        \bsj {(\mathtt{app}\;{e_1}\;{e_2}, \downarrow)} {(v, \delta')}
    }
    
    \inferrule[PM-App2]{
        \val {v_1}
        \\
        \bsj {({e_2}, \downarrow)} {(v_2, \delta)}
        \\
        \bsj {(\mathtt{app3}\;{v_1}\;{o_2}, \delta)} {(v, \delta')}
    }{
        \bsj {(\mathtt{app2}\;{v_1}\;{e_2}, \downarrow)} {(v, \delta')}
    }
    
    \inferrule[PM-App3]{
        \bsj {[\lam {f} {x} e/f, v/x]e, \downarrow} {(v, \delta)}
    }{
        \bsj {(\mathtt{app3}\;{\lam {f} {x} e}\;{v}, \downarrow)} {(v, \delta)}
    }
    \end{mathpar}
    \caption{Flag-based \cite{poulsen2017flag}
    inductive/coinductive rules for application with divergence}
    \label{fig:flagrules}
\end{figure}

\paragraph{Inductive Techniques}


Unlike all of other work 
discussed so far,
Gunter and R\'emy's
partial proof semantics \cite{gunter1993proof}
does not use a coinductive relation 
to handle divergence. Instead,
partial proofs 
augment big-step 
semantics with the ability 
to abstractly denote 
values with ``logic variables'' 
and leave the derivation of 
those variables uncalculated.
This approach has the effect 
of allowing big-step derivations 
to exist in a partially-completed 
form, which means nonterminating computation 
can be represented with a sequence of
arbitrary-depth derivations.
The rules of this system 
use two forms of big-step evaluation 
judgments, the ``search'' judgment 
$\ppj e u$ and the ``redex''
judgment $\bsj {e} {u}$,
to help ensure that one cannot 
make arbitrarily-deep derivations 
of terminating computations.
Each syntactic form in this 
system requires its own search and 
redex rules.
The resulting rules look like 
those of \Cref{fig:partialproofrules},
where $v$ ranges over only 
values, $z$ ranges over only logic 
variables, and $u$ ranges over both.
This rule set shows the
two rules for application as
well as the two rules for 
stopping the growth of 
the proof tree, 
\textsc{GR-StopS} and 
\textsc{GR-StopR}.

\begin{figure}
    \small
    \begin{mathpar}
        \inferrule[GR-StopS]{
            }{
                \ppj {e} z
            }
    
            \inferrule[GR-StopR]{
            }{
                \bsj {e} z
            }

        \inferrule[GR-AppS]{
            \ppj {e_1} {u_1}    
            \\
            \ppj {e_2} {u_2}
            \\
            \bsj {\app {u_1} {u_2}}  {u}
        }{
            \ppj{\app {e_1}{e_2}} u
        }

        \inferrule[GR-AppR]{
            \ppj {[\lam {f} {x} e/f, v/x]e} u
        }{
            \bsj {\app {\lam {f} {x} e}{v}} u
        }
    \end{mathpar}
    \caption{Partial-proof \cite{gunter1993proof}
    inductive rules for application}
    \label{fig:partialproofrules}
\end{figure}

Another inductive approach comes from 
a line of work based on 
step-indexing or 
``fuel'' \cite{ernst2006virtual,siek2012,amin2017type}. 
This approach 
augments the operational semantics with a counter
and augments the language with time-out error states.
The counter counts down in a manner
corresponding to steps or derivation depth. 
When the counter 
is above 0, the semantics behave as normal.
However, when the counter hits 0, 
the evaluation is forced into a time-out 
error state represented here by $\nont$. By quantifying over 
counters, this approach allows for
every finite prefix of computation to be 
captured in a big-step manner (although 
each strict prefix concludes in an error state).
The rules of such a fuel-based system 
are exemplified in \Cref{fig:stepind},
where $e \Downarrow_n v$ means that 
$e$ evaluates to $v$ with a big-step 
derivation no deeper than $n$,
and where $\nont$
represents the error state. 
These fuel-based rules structurally 
match those for 
the coinductive interpretation 
of divergence (\Cref{fig:cousotrules})
except that here they are inductive
and they include the rule
\textsc{SC-Stop}. A typical type safety theorem using fuel 
would state 
that, given any amount of fuel, 
a well-typed expression evaluates to 
either a value of the same type or an error state,
$\forall n.\, \cdot \vdash e:\tau \wedge e \Downarrow_n v 
\implies \cdot \vdash v:\tau \vee v = \nont$.
The idea behind such a theorem is that,
when enough fuel is given, the values 
of terminating computations are captured,
and when too little fuel is given (including for 
infinite computations),
it is still verified that evaluation could 
safely occur up until the fuel is used up and 
the time-out error is reached.

\begin{figure}
    \small
\begin{mathpar}
\inferrule[SC-Stop]{
}{
    e \Downarrow_0 \nont
}

\inferrule[SC-App1]{
    {e_1} \Downarrow_{n-1} \nont 
}{
    {\app {e_1}{e_2}} \Downarrow_n \nont 
}

\inferrule[SC-App2]{
    {e_1} \Downarrow_{n-1} {v_1} 
    \\
    {e_2} \Downarrow_{n-1} \nont 
}{
    {\app {e_1}{e_2}} \Downarrow_{n} \nont 
}

\inferrule[SC-App3]{
    {e_1} \Downarrow_{n-1} {\lam {f} {x} e}  
    \\
     {e_2}  \Downarrow_{n-1} {v} 
    \\\\
     {[\lam {f} {x} e/f, v/x]e} \Downarrow_{n-1} \nont
}{
    {\app {e_1}{e_2}} \Downarrow_{n} \nont 
}
\end{mathpar}
\caption{Counter-based \cite{ernst2006virtual,siek2012,amin2017type} inductive rules 
for application}
\label{fig:stepind}
\end{figure}

While fuel counters are inductive, it is also possible 
to adapt them to work alongside coinductive rules like in the coevaluative
work of Z\'u\~niga and Bel-Enguix \cite{zuniga2022coevaluation}.
This work has two sets of rules: one set for concluding with values
and one for concluding with nontermination. 
The counters count down in the value rules to ensure 
that, even though they are interpreted coinductively, 
these rules can only be 
used for terminating computations---an infinite 
computation would run down any counter, preventing these rules' use.
Thus, nonterminating computations cannot coevaluate to values here, 
unlike in other coevaluative systems.
Simultaneously, the nontermination rules (\Cref{fig:zuniga}) 
ignore the counters and thus act 
like the standard rules of Cousot and Cousot (\Cref{fig:cousotrules}).

\begin{figure}
    \small
\begin{mathpar}
\inferrule[ZB-App1]{
    {e_1} \Downarrow_{m} \nontinf
}{
    {\app {e_1}{e_2}} \Downarrow_n \nontinf
}

\inferrule[ZB-App2]{
    {e_1} \Downarrow_{\ell} {v_1} 
    \\
    {e_2} \Downarrow_{m} \nontinf
}{
    {\app {e_1}{e_2}} \Downarrow_{n} \nontinf
}

\inferrule[ZB-App3]{
    {e_1} \Downarrow_{k} {\lam {f} {x} e}  
    \\
     {e_2}  \Downarrow_{\ell} {v} 
    \\\\
     {[\lam {f} {x} e/f, v/x]e} \Downarrow_{m} \nontinf
}{
    {\app {e_1}{e_2}} \Downarrow_{n} \nontinf
}
\end{mathpar}
\caption{Counter-based \cite{zuniga2022coevaluation} coinductive rules 
for application}
\label{fig:zuniga}
\end{figure}

Hoffmann and Hofmann's partial 
big-step semantics tackle nontermination inductively,
without coinduction
\cite{hoffmann2010amortized}.
While their work 
is specialized to reasoning about 
program cost, we abstract 
to its key takeaways here.
Partial big-step semantics
handles nontermination by 
introducing a new judgment, 
$\bsej e \nont a$, to represent the 
nondeterministic halting of computation
during $e$'s evaluation after accumulating 
some trace $a$. (Note that here the symbol 
$\nont$ is \textit{not} an expression and 
is considered part of the judgment notation.)
This approach behaves similarly to 
the step-counting approach without needing 
any counter.
An example of partial big-step 
rules for function application can 
be found in \Cref{fig:partialbigstep}.\footnote{
    Technically the work on 
    partial big-step semantics also 
    uses a variant of monadic normal form. 
    However, this form plays no role 
    in their partial big-step semantics,
    and is instead used for 
    other conveniences related to 
    rule presentation.
    As monadic normal form 
    would obscure some of the 
    comparisons made in this section, 
    we do not use it for this example.
}
Big-stop semantics is based on this work.

\begin{figure}
    \small
    \begin{mathpar}
\inferrule[HH-Stop]{
}{
    \bsej e \nont 1
}

\inferrule[HH-App1]{
    \bsej {e_1} \nont a
}{
    \bsej {\app {e_1}{e_2}} \nont a
}

\inferrule[HH-App2]{
    \bsej {e_1} {v_1} a  
    \\
    \bsej {e_2} \nont b
}{
    \bsej {\app {e_1}{e_2}} \nont {ab}
}

\inferrule[HH-App3]{
    \bsej {e_1} {\lam {f} {x} e}  a
    \\
    \bsej {e_2} {v} b
    \\
    \bsej {[\lam {f} {x} e/f, v/x]e} \nont c
}{
    \bsej {\app {e_1}{e_2}} \nont {abc}
}
    \end{mathpar}
    \caption{Hoffmann-Hofmann-style partial big-step \cite{hoffmann2010amortized}
    inductive rules for application}
    \label{fig:partialbigstep}
\end{figure}

Later work by D. M.
Kahn improved upon the partial 
big-step approach to cost analysis by introducing 
an algebraic annihilator for 
cost traces and treating
$\nont$ as a value
to reduce propagation rules \cite{kahn2024leveraging}.
The resulting system only adds one new
rule to the standard 
big-step rules, rather than 
a few rules for each syntactic form,
and it makes use of a 
variant of monadic normal form to 
help with this rule minimization.
An example of Kahn's rules for 
application can be found in 
\Cref{fig:kbps}. This work formed the basis 
of some of the ergonomic optimizations
of \Cref{sec:opt}.

\begin{figure}
    \small
    \begin{mathpar}
\inferrule[K-Stop]{
}{
    \bsej e \nont 0
}

\inferrule[K-App]{
    \bsej {[\lam {f} {x} e/f, v/x]e} {v'} a
}{
    \bsej {\app {\lam {f} {x} e}{v}} {v'} {a}
}
    \end{mathpar}
    \caption{Kahn-style partial big-step \cite{kahn2024leveraging}
    inductive rules for application with divergence}
    \label{fig:kbps}
\end{figure}

\paragraph{Comparison with Big-Stop}

When comparing big-stop semantics 
and the other systems 
discussed here, 
there are clear similarities. Aside 
from the work that big-stop is based on, 
one of the most similar lines of work is 
that using fuel \cite{ernst2006virtual,siek2012,amin2017type}.
By consuming fuel in proportion to 
evaluation step count, one should be able to induce 
an error state at exactly the same expressions 
that big-stop can stop. If the error states are 
then altered to record the last pre-error
expression, then such an approach would successfully 
match small-step while using
inductive, big-step-style rules, just as big-stop does. 
However, big-stop can 
get these same results with less work.
Big-stop rules are essentially a 
superset of big-step rules, whereas fuel-based rules 
are more intrusive to implement because every 
expression must track additional state: the fuel counter. 
This state 
must be quantified over in fuel-based formalizations, 
and such quantification must be 
discharged by, e.g., proving that program evaluation 
is sufficiently independent of fuel count.
Because big-stop needs none of this, 
it shows that the fuel of 
fuel-based approaches is actually unnecessary.
Compare the typical fuel-based type safety statement
given earlier in this section to big-stop's
\Cref{thm:pres}.

Otherwise, the effectful version of big-stop 
semantics has the most obvious 
similarities to discuss.
Structurally,
the inductive judgment for effectful
stopping is very 
similar to the coinductive 
judgment for divergence \cite{cousot1992inductive}. 
Like the use of the 
flag $\downarrow$ from flag-based 
semantics \cite{poulsen2017flag},
annihilator-based big-stop semantics
uses the special effect $0$
to prevent spurious effects from 
spurious evaluations that are introduced 
in the effort to reduce 
rule count. 
Like partial proof semantics \cite{gunter1993proof},
big-stop derivations can 
be cut off at any point.
And of course,
the effectful big-stop semantics 
is based on partial big-step semantics 
\cite{hoffmann2010amortized,kahn2024leveraging}.

Nonetheless, big-stop semantics
introduces some clear benefits over 
most other systems discussed here.
One benefit is that 
it is inductive, which 
coheres better with other 
inductively defined judgments like
typing, and which enables 
induction-based proofs 
rather than less-commonly-understood 
coinduction-based proofs.
Another benefit is that it 
can be optimized 
to require extremely few rules---only 
one or two rules are needed for 
entire languages, whereas other 
approaches require multiple 
rules for each syntactic form.
It also does not introduce additional complexity
like counters to the judgments.
And finally, big-stop semantics 
has direct correspondence with 
the behaviour of 
small-step semantics,
whereas coevaluation 
has a less-clear relationship.

\paragraph{Other Related Work}
Much other work exists which is less 
comparable to big-stop semantics,
and almost all of it is coinductive.
A few such examples are listed as follows:
Nakata and Uustalu introduce 
a mutually-coinductive 
trace-based semantics for an 
imperative language
that includes infinite loops
\cite{nakata2009trace}.
Capretta's delay monad
uses coinductive types to 
capture partial computation in a type-theoretic 
way
\cite{capretta2005general}, and Danielsson has 
applied this monad to operational 
semantics, resulting in 
something similar to a 
hybrid big-step operational/denotational semantics
that can handle nontermination
via giving a big-step system 
a mixed recursive/corecursive 
definitional interpreter 
equipped with a partiality monad
\cite{danielsson2012operational}.
Dagnino uses coaxioms 
to develop a metatheory 
for divergent reasoning 
concerning big-step systems.
\cite{dagnino2022meta}.
A similar approach has also been 
applied to the effectful setting,
where infinite sequences 
of effects are formalized as an 
$\omega$-monoid
\cite{ancona2020big}, rather than 
just a finite monoid as in 
this work.


\section{Conclusion}\label{sec:conc}

Big- and small-step operational semantics 
are popular forms of semantics that 
traditionally have occupied subtly 
different niches. Small-step 
semantics describe each individual step 
of computation, and can capture 
diverging computation just as well as 
converging. Big-step semantics 
sacrifice the ability to reason about 
diverging computation in exchange 
for easy access to the 
values resulting from computation,
which is more convenient for purposes such as proving semantic
preservation.

Big-stop semantics recaptures 
the expressive power of small-step 
semantics with all the ergonomic benefits
of big-step semantics. 
Big-stop semantics works by
introducing rules for
nondeterministically stopping computation. 
In this article, we have presented big-stop semantics for a
call-by-value variant of PCF and a simple imperative language.
However, we are not aware of obstacles that prevent or complicate the
adaptation of big-stop semantics to more complex languages that enjoy
big-step semantics.
Nevertheless, it remains an open research question if big-stop semantics
can be extended to language features, such as concurrency and
continuations, that are challenging to express in a big-step
semantics.

Big-stop semantics is the latest in a 
long line of work that aims to 
regain the power of small-step 
semantics in a big-step-style system. 
However, other work 
in this vein has focused on coinductive techniques
and/or required comparatively large changes 
to the standard big-step system.
For this reason, the simplicity of 
big-stop semantics is notable.
Big-stop semantics is both 
inductive and can be formulated
to require very few changes 
to the original big-step system.



\begin{acks}
  We would like to first acknowledge Martin Hofmann (1965-2018),
  whose initial idea of replacing coinductive techniques for handling
  nontermination in a cost analysis setting resulted in the partial
  big-step semantics upon which big-stop semantics was built.
  We would also like to acknowledge Bob Harper and the anonymous
  reviewers of this article for their helpful feedback on this work.

  This material is based upon work supported by the
  \grantsponsor{NSF}{National Science
    Foundation}{https://doi.org/10.13039/100000001} under Grants
  Nos.~\grantnum{NSF}{2311983} and ~\grantnum{NSF}{2525102}, and the 
  \grantsponsor{AFOSR}{United States Air Force Office of Scientific Research}{}
  under grant number \grantnum{AFOSR}{FA9550-21-0009} and 
  \grantnum{AFOSR}{FA9550-23-1-0434} (Tristan Nguyen, program manager). 
  Any opinions, findings, and conclusions or recommendations in this
  material are those of the authors and do not necessarily reflect the
  views of the NSF or the AFOSR.
  
\end{acks}

\bibliographystyle{ACM-Reference-Format}
\bibliography{sources}

\clearpage 
\appendix

\section{Types}\label{sec:types}

\begin{figure}
  \small
    \begin{mathpar}
        %
        %
        \inferrule[T-Var]{
        }{
            \ty {\Gamma,x:\tau} x \tau
        }

        \inferrule[T-Lam]{
            \ty {\Gamma, f:\fun \tau \sigma, x:\tau} e \sigma
        }{
            \ty {\Gamma} {\lam {f} {x} e} {\fun \tau \sigma}
        }

        \inferrule[T-Zero]{
        }{
            \ty \Gamma \z \nat
        }

        \inferrule[T-Succ]{
            \ty \Gamma e \nat
        }{
            \ty \Gamma {\suc e} {\nat}
        }

        \inferrule[T-App]{
            \ty \Gamma {e_1} {\fun \sigma \tau} 
            \\
            \ty \Gamma {e_2} {\sigma}
        }{
            \ty \Gamma {\app {e_1} {e_2}} \tau
        }

        \inferrule[T-Case]{ 
            \ty \Gamma {e_3} \nat
            \\
            \ty \Gamma {e_1} \tau 
            \\\\
            \ty {\Gamma, x:\nat} {e_2} \tau
        }{
            \ty \Gamma {\ifz {e_1} x {e_2} {e_3}} \tau 
        }

        \inferrule[T-Eff]{
            \ty {\Gamma} e \tau
        }{
            \ty {\Gamma} {\eff a e} \tau
        }
    \end{mathpar}
    \caption{Typing rules for PCF}
\label{fig:ty}
\end{figure}

This section contains types and typing rules for PCF.

The only types of PCF are
natural numbers and functions, formalized by the following grammar:
\[ \tau ::= \nat \mid \fun {\tau_1} {\tau_2} \]

The typign rules are
given 
in \Cref{fig:ty}.
The typing rules assign
types to expressions using the
judgment $\ty{\Gamma}{e}{\tau}$, which
means that the expression $e$ has type $\tau$ 
given the typing assumptions of the typing context $\Gamma$.

The typing rules are standard, with the exception of \textsc{T-Eff},
which is not used in the current pure setting. (It is relevant 
to the effectful setting of \Cref{sec:eff}.)
Note that types of expressions may not be 
unique since functions exist like the identity $\lam \_ x x$ 
which can be typed as $\fun \tau \tau$ for any $\tau$.

\section{Big-Stop Examples}\label{sec:examples}

Here we collect some example derivations using big-stop.

\subsection{Pure}

To illustrate some big-stop derivations 
and show how big-stop operational 
semantics differ from an alternative 
approach based on ``coevaluation'',
we consider two expressions extracted
from the work of Leroy and Grall \cite{leroy2009coinductive}.
(For more about this system, see \cref{sec:related}.)

\subsubsection{Pure Example 1}
First
we consider the following expression of type $\fun \nat \nat$:
\[{\app {\lam \_ {x} \z} {\app {\lam f y {\app f y}} \z}}\] 

This expression is our language's well-typed
version of the expression $(\lambda x.\, 0) \omega$,
which Leroy and Grall's system coevaluates to 0 \cite{leroy2009coinductive}.
In typical eager small-step semantics,
however, this expression never 
reduces to a value. Instead, it
loops forever 
while attempting to evaluate the argument 
$\app {\lam f y {\app f y}} \z$.

The big-stop system captures the looping behaviour 
of the expression in the following ways.
Firstly, it allows for a
trivial
derivation that the expression big-stops to itself 
via applying 
\textsc{St-Stop(0)}.
%
%
Secondly, it is able to conclude 
the same big-stop judgment nontrivially using a progressing 
derivation, as guaranteed by \Cref{thm:progress}.
To this end, let $e = {\lam f y  {\app f y}}$.
(For space, we elide premisses of the form $\val v$.)

\begin{small}
\begin{mathpar}
    \inferrule*[Right=\textsc{St-Stop(2)},leftskip=1cm]{
        \inferrule*[Right=\textsc{St-Stop(0)},rightskip=-1cm]{
        }{
            \lam \_ {x} \z
        }
        \\
        \inferrule*[Right=\textsc{St-App},vdots=1cm,leftskip=3cm]{
            \inferrule*[Right=\textsc{St-Stop(0)}]{
            }{
                \bstj {e} {e}
            }
            \\
            \inferrule*[Right=\textsc{St-Stop(0)},leftskip=-1cm,rightskip=-1cm]{
            }{
                \bstj \z \z
            }
            \\
           \inferrule*{
            }{
                \val \z
            }
            \\
            \inferrule*[Right=\textsc{St-Stop(0)}]{
            }{
                \bstj {\app e \z} {\app e \z}
            }
        }{
            \bstj {\app e \z} {\app e \z}
        }
        \\
        \inferrule*[leftskip=2cm]{
        }{
            \val {\lam \_ {x} \z}
        }
    }{
        \bstj {\app {\lam \_ {x} \z} {\app e \z}}{\app {\lam \_ {x} \z}  {\app e \z}}
    }
\end{mathpar}
\end{small}

In fact, any number of applications of the progressing
rule \textsc{St-App} may be applied, as the upper-right application 
of \textsc{St-Stop(0)} has the same conclusion as
the whole \textsc{St-App} subtree.
This circumstance shows that 
any finite amount of computation of this expression 
can be appropriately expressed in the big-stop system.
In contrast, coevaluation captures the 
``infinite'' computation that finishes 
the infinite looping of the argument and then tosses away that value
using the constant function $\lam \_ { x} \z$,
leaving the result of $\z$. In effect,
this behaviour of coevaluation seems to ignore the distinction between call-by-value 
and call-by-name evaluation order.

\subsubsection{Pure Example 2}
The second example is the following expression 
of type $\fun \nat \nat$ found by 
A. Filinski \cite{leroy2009coinductive}.
\[ 
    {\app {\lam f x {{\lam \_ {y} {\app {\app {f} x} y} }}} \z}
\]
This expression does not coevaluate at all in Leroy and Grall's system, 
but it does have
big-stop derivations. Letting
$e'$ be $\lam f x {{\lam \_ {y} {\app {\app {f} x} y} }}$,
the following 
is such a progressing derivation.
This derivation can also be extended to be arbitrarily deep.

\begin{mathpar}
\inferrule*[Right=\textsc{St-App},leftskip=1cm]{
    \inferrule*[Right=\textsc{St-Stop(0)},rightskip=-1cm]{
    }{
        \bstj {e'} {e'}
    }
    \\
    \inferrule*[Right=\textsc{St-Stop(0)},rightskip=-1cm]{
    }{
        \bstj {\z} {\z}
    }
    \\
    \inferrule*{
    }{
        \val \z
    }
    \\
    \inferrule*[Right=\textsc{St-Stop(0)},vdots=1cm,leftskip=4cm]{
    }{
        \bstj {\lam \_ y {\app {\app {e'} \z} {y}}}  {\lam \_ y {\app {\app {e'} \z} {y}}}
    }
}{
    \bstj {\app {e'} \z}  {\lam \_ y {\app {\app {e'} \z} {y}}}
}
\end{mathpar}

\subsection{Effectful}

\subsubsection{Application to Resource Analysis 1}

To demonstrate the use of 
the effectful big-stop system, consider 
reasoning about the memory requirements of 
some function $f$. One might care 
to ensure that $f$'s application does 
not require more memory than is available,
say, 2 cells of memory.
One could approach this task with effects by using
the effect ``$\mathit{alloc}$'' to represent 
the allocation of one new cell of memory.
Then the derivation of a judgment like 
$\msej {\app f v} {e} {(\mathit{alloc})^2}$
would show that exactly 2 memory cells are required to 
reduce to $e$ given argument $v$, and the desired
memory safety property would be formalized 
with the statement 
$\msej {\app f v} {e} {(\mathit{alloc})^n} \implies n \leq 2$.
In this way, effects can be more useful to 
consider than the actual results of computation.

Using only effectful big-step (not stop) semantics, 
a similar memory-safety property can be 
stated: for all values $v : \nat$, 
$\bsej {\app f v} {v'} {(\mathit{alloc})^n} \implies n \leq 2$.
However, this new statement only captures 
the memory behaviour of terminating computations. 
It misses nonterminating behaviour of
functions like the following:
\[
    \lam {f} {x} {{\ifz {\z}{y} {\app {\lam {g} z {\eff {\mathit{alloc}} {\app g z}}} \z} {x}}}
\]

Depending upon whether the argument to the function $f$ is zero,
$f$ either allocates
no memory before immediately 
terminating
or allocates infinite memory while looping forever. 
The latter behaviour will clearly 
surpass 5 cells and thus is unsafe. 
However, the big-step memory-safety statement
$\bsej {\app f v} {v'} {(\mathit{alloc})^n} \implies n \leq 2$
is still true
because the infinite looping behaviour never 
results in a value $v'$. Big-step semantics is unsuited to 
formalizing the property of interest.

Using big-stop semantics, however, 
the property
$\bsej {\app f v} {e} {(\mathit{alloc})^n} \implies n \leq 2$
actually expresses what we want.
Moreover, this statement can be shown 
false with the following derivation scheme capable of 
inferring $\bsej {\app f v} {e} {(\mathit{alloc})^n}$ for any $n$,
where $v= \suc \z$, $e = \app {e'} \z$, and $e' = \lam {g} {z} {\eff {\mathit{alloc}} {\app g z}}$.

\begin{small}
\begin{mathpar}
\inferrule*[Right=\textsc{StE-App},leftskip=1cm]{
    \inferrule*[Right=\textsc{StE-Stop(0)}]{
    }{
        \bstej f f 1
    }
    \\
    \inferrule*[Right=\textsc{StE-Stop(0)},vdots=1cm]{
    }{
        \bstej {\suc \z} {\suc \z} 1
    }
    \\
    \inferrule*[leftskip=1cm]{
    }{
        \val {\suc \z} 
    }
    \\
    \inferrule*[Right=\textsc{StE-CaseS},vdots=2cm,leftskip=3cm]{
        \inferrule*[Right=\textsc{StE-Stop(0)}]{
        }{
            \bstej {\suc \z} {\suc \z} 1
        }
        \\
        \inferrule*[leftskip=-1.5cm]{
        }{
            \val {\suc \z} 
        }
        \\
        \mathcal{D}(n)
    }{
        \bstej {\ifz \z y {\app {e'} \z} {\suc \z}} {\app {e'} \z} {(\mathit{alloc})^n}
    }
}{
    \bstej {\app {\lam f x {\ifz \z y {\app {e'} \z} x}} {\suc \z}} {\app {e'} \z} {(\mathit{alloc})^n}
}
\end{mathpar}
\end{small}

where for $n>0$, $\mathcal{D}(n)$ is given by the following 
derivation
\begin{small}
\begin{mathpar}
    \inferrule*[Right=\textsc{StE-App},leftskip=1cm]{
            \inferrule*[Right=\textsc{StE-Stop(0)},rightskip=-0.5cm]{
            }{
                \bstej {e'} {e'} 1
            }
            \\
            \inferrule*[Right=\textsc{StE-Stop(0)},vdots=1cm]{
            }{
                \bstej \z \z 1
            }
            \\
            \inferrule*[leftskip=0.5cm]{
            }{
                \val { \z} 
            }
            \\
            \inferrule*[Right=\textsc{StE-Eff},leftskip=1.5cm,vdots=2cm,rightskip=1.5cm]{
                \mathcal D (n-1)
            }{
                \bstej {\eff {\mathit{alloc}} {\app {e'} \z}} {\app {e'} \z} {(\mathit{alloc})^n}
            }
        }{
            \bstej {\app {e'} \z} {\app {e'} \z} {(\mathit{alloc})^n}
        }
\end{mathpar}
\end{small}

and $\mathcal{D}(0)$ is the following derivation.
\begin{small}
\begin{mathpar}
\inferrule[StE-Stop(0)]{
}{
    \bstej {\app {e'} \z} {\app {e'} \z} 1
}
\end{mathpar}
\end{small}

In this way, big-stop semantics allows 
one to reason correctly about nonterminating
behaviour.

\subsubsection{Application to Resource Analysis 2}
Now consider verifying that memory requirement of 
the following function $f$ does not exceed 2 cells 
of memory on any input value $v:\nat$.
\[
    \lam {f} {x} {\eff {\mathit{alloc}}
    {{\ifz {\z}{y}{\app {\lam {g} {z} {{\app g z}}} \z} {x}}}}
\]

Unlike the previous example, this function actually 
does stay under 2 cells of memory no matter the input;
it is true that 
$\bsej {\app {f} v} {v'} {(\mathit{alloc})^n} \implies n \leq 2$
for
all values $v : \nat$.
In fact, the function only ever requires one allocation.

However, big-step semantics is unable to properly verify this 
bound, again due to being unable to reason about the nonterminating 
execution that occurs when the input $v$ is nonzero. 
Big-\textit{stop} semantics, on the other 
hand, can verify this bound without issue. Such a proof 
would look something like the following, where 
$\omega = {\app {\lam {g} {z} {{\app g z}}} \z}$.

Suppose $v = \z$. Then the following derivation shows 
that only 1 allocation is required to reach the value $\z$. 

\begin{small}
\begin{mathpar}
\inferrule*[Right=\textsc{StE-App},leftskip=1cm]{
    \inferrule*[Right=\textsc{StE-Stop(0)}]{
    }{
        \bstej f f 1
    }
    \\
    \inferrule*[Right=\textsc{StE-Stop(0)},vdots=1cm,leftskip=-1cm]{
    }{
        \bstej \z \z 1
    }
    \\
    \inferrule*[leftskip=0.5cm]{
    }{
        \val \z
    }
    \\
    \inferrule*[Right=\textsc{StE-Eff},vdots=2cm,leftskip=2cm]{
        \inferrule*[Right=\textsc{StE-CaseZ}]{
            \inferrule*[Right=\textsc{StE-Stop(0)}]{
            }{
                 \bstej \z \z 1
            }
            \\
            \inferrule*[Right=\textsc{StE-Stop(0)},leftskip=-1.5cm]{
            }{
                 \bstej \z \z 1
            }
        }{
            \bstej 
            {{{\ifz {\z}{y}{\omega} {\z}}}}
            \z 
            {1}
        }
    }{
        \bstej 
        {\eff {\mathit{alloc}}{\ifz {\z}{y}{\omega}}}
        \z 
         {\mathit{alloc}}
    }
}{
    \bstej {\app f \z} \z {\mathit{alloc}}
}
\end{mathpar}
\end{small}

Now suppose instead that $v = \suc{v''}$ for some value $v'':\nat$.
Observe that arbitrarily deep progressing derivations can be 
made for which no more than 1 allocation is required. 
Deep enough derivations take the following form:

\begin{small}
\begin{mathpar}
\inferrule*[Right=\textsc{StE-App},leftskip=1cm]{
    \inferrule*[Right=\textsc{StE-Stop(0)}]{
    }{
        \bstej f f 1
    }
    \\
    \inferrule*[Right=\textsc{StE-Stop(0)},vdots=1cm,leftskip=-0.5cm]{
    }{
        \bstej {\suc {v''}} {\suc {v''}} 1
    }
    \\
    \inferrule*[leftskip=1cm]{
    }{
        \val {\suc {v''}}
    }
    \\
    \inferrule*[Right=\textsc{StE-Eff},vdots=2cm,leftskip=3cm]{
        \inferrule*[Right=\textsc{StE-CaseS}]{
            \inferrule*[Right=\textsc{StE-Stop(0)},rightskip=-1cm]{
            }{
                  \bstej {\suc {v''}} {\suc {v''}} 1
            }
            \\
            \mathcal D(n)
        }{
            \bstej 
            {{{\ifz {\z}{y}{\omega} {{\suc {v''}}}}}}
            {\omega}
            {1}
        }
    }{
        \bstej 
        {\eff {\mathit{alloc}}{{\ifz {\z}{y}{\omega} {{\suc {v''}}}}}}
        {\omega}
         {\mathit{alloc}}
    }
}{
    \bstej {\app f {\suc {v''}}} {\omega} {\mathit{alloc}}
}
\end{mathpar}
\end{small}

where for $n>0$, $\mathcal{D}(n)$ is given by the following 
derivation
\begin{small}
\begin{mathpar}
\inferrule*[Right=\textsc{StE-App},leftskip=1cm]{
    \inferrule*[Right=\textsc{StE-Stop(0)}]{
    }{
        \bstej {\lam g z {\app g z}} {\lam g z {\app g z}} 1
    }
    \\
    \inferrule*[Right=\textsc{StE-Stop(0)},vdots=1cm,leftskip=-1cm]{
    }{
        \bstej {\z} {\z} 1
    }
    \\
    \mathcal{D}(n-1)
}{
    \bstej {\omega} {\omega} 1
}
\end{mathpar}
\end{small}

and $\mathcal{D}(0)$ is the following 
derivation.
\begin{small}
\begin{mathpar}
\inferrule*[Right=\textsc{StE-Stop(0)},leftskip=1cm]{
}{
    \bstej {\omega} {\omega} 1
}
\end{mathpar}
\end{small}

The only other remaining big-stop
derivations arise from 
using \textsc{StE-Stop} to cut the 
previously given derivations short.
In any case, however, such shorter derivations can 
require no more 
allocations than the deeper derivations.

This proof sketch also demonstrates how big-stop semantics 
can make use of both big-step and small-step reasoning techniques 
when convenient. The first case used typical 
big-step style reasoning to jump straight to the resulting value
with no fuss. The second case used typical small-step 
style reasoning to show an evaluation loop. The big-stop 
system is ergnomically 
convenient because it has access to both kinds of reasoning techniques.

\section{Additional Imperative Formalism}\label{sec:extraimp}

This section contains some extra semantic 
formalizations related to the imperative language 
from \Cref{sec:imp}.

\begin{figure}
     \small
 \begin{mathpar}
 \inferrule[BI2-Bind]{
     \sigma_1(a) = z
 }{
     \bsj {\langle {\iset x a} \mid \sigma \rangle}  {[z/x]\sigma}
 }

 \inferrule[BI2-Seq]{
     \sigma_1(a) = 0
     \\
     \bsj  {\langle s_1  \mid {\sigma_1} \rangle} {\sigma_2} 
     \\
     \bsj  {\langle s_2  \mid {\sigma_2} \rangle}  {\sigma_3}  
 }{
     \bsj {\langle {\iseq {s_1} {s_2}}  \mid {\sigma_1} \rangle} \sigma_3 
 }

 \inferrule[BI2-Skip]{
 }{
     \bsj {\langle \iskip \mid \sigma \rangle}  \sigma 
 }

 \inferrule[BI2-Then]{
     \sigma(a) \neq 0
     \\
     \bsj {\langle {s} \mid {\sigma} \rangle}  \sigma'
 }{
     \bsj {\langle {\iif a {s}} \mid {\sigma} \rangle}  \sigma' 
 }

 \inferrule[BI2-Else]{
     \sigma_1(a) = 0
 }{
     \bsj {\langle {\iif a {s}}  \mid {\sigma} \rangle}  \sigma 
 }

 \inferrule[BI2-Done]{
     \sigma_1(a) = 0
 }{
     \bsj {\langle {\iwhile a {s}} \mid {\sigma} \rangle}  \sigma 
 }

 \inferrule[BI2-Do]{
     \sigma_1(a) \neq 0
     \\
     \bsj {\langle {s} \mid {\sigma_1} \rangle} \sigma_2 
     \\
     \bsj {\langle  \iwhile a {s} \mid {\sigma_2} \rangle} \sigma_3 
 }{
     \bsj {\langle  \iwhile a {s} \mid {\sigma_1} \rangle} \sigma_3
 }
 \end{mathpar}
 \caption{Big-step semantics for the imperative language}
 \label{fig:altbig}
\end{figure}

\begin{figure}
     \small
 \begin{mathpar}
 \inferrule[StI2-Bind]{
     \sigma_1(a) = z
 }{
     \bstj {\langle {\iset x a} \mid \sigma \rangle}  {[z/x]\sigma}
 }

 \inferrule[StI2-Seq]{
     \sigma_1(a) = 0
     \\
     \bstj  {\langle s_1  \mid {\sigma_1} \rangle} {\sigma_2} 
     \\
     \bstj  {\langle s_2  \mid {\sigma_2} \rangle}  {\sigma_3}  
 }{
     \bstj {\langle {\iseq {s_1} {s_2}}  \mid {\sigma_1} \rangle} \sigma_3 
 }

 \inferrule[StI2-Skip]{
 }{
     \bstj {\langle \iskip \mid \sigma \rangle}  \sigma 
 }

 \inferrule[StI2-Then]{
     \sigma(a) \neq 0
     \\
     \bstj {\langle {s} \mid {\sigma} \rangle}  \sigma'
 }{
     \bstj {\langle {\iif a {s}} \mid {\sigma} \rangle}  \sigma' 
 }

 \inferrule[StI2-Else]{
     \sigma_1(a) = 0
 }{
     \bstj {\langle {\iif a {s}}  \mid {\sigma} \rangle}  \sigma 
 }

 \inferrule[StI2-Done]{
     \sigma_1(a) = 0
 }{
     \bstj {\langle {\iwhile a {s}} \mid {\sigma} \rangle}  \sigma 
 }

 \inferrule[StI2-Do]{
     \sigma_1(a) \neq 0
     \\
     \bstj {\langle {s} \mid {\sigma_1} \rangle} \sigma_2 
     \\
     \bstj {\langle  \iwhile a {s} \mid {\sigma_2} \rangle} \sigma_3 
 }{
     \bstj {\langle  \iwhile a {s} \mid {\sigma_1} \rangle} \sigma_3
 }

 \inferrule[StI2-Stop]{
 }{
     \bstj {\langle s \mid \sigma \rangle} {\mathtt{freeze}(\sigma)}
 }
 \end{mathpar}
 \caption{Alternative big-stop semantics for the imperative language}
 \label{fig:altstop}
\end{figure}


Consider the typical big-step judgment of the form 
$\bsj {\langle s \mid \sigma \rangle} {\sigma'}$.
This form of judgment is slightly trickier to adapt to big-stop
because big-stop leverages judgements 
of the form $\bstj {\langle s \mid \sigma \rangle} {\langle \iskip \mid \sigma' \rangle}$
to distinguish total from partial evaluations. 
For example, the use of $\iskip$ distinguishes the applicability of
\textsc{BI-Do} from that of
\textsc{StI-Do}. 

Nonetheless, this problem can be overcome with the 
annihilator trick of \Cref{sec:opt}. An appropriate 
new annihilator effect is given by the function $\mathtt{freeze}$.
This function renders a state immutable so that 
$[z/x]\mathtt{freeze}(\sigma) =\mathtt{freeze}(\sigma)$.
The result of this approach is the judgment
\[
 \bstj  {\langle s  \mid {\sigma} \rangle}  {\sigma'}
\]
defined by the rules of \Cref{fig:altstop}.
These rules are comprised of the
standard big-step rules and the additional rule \textsc{StI2-Stop}.
\begin{mathpar}
  \inferrule[StI2-Stop]{
}{ \bstj {\langle s \mid \sigma \rangle} {\mathtt{freeze}(\sigma)}}
\end{mathpar}
These alternative big-stop semantics agree 
with the small-step semantics in the following way:

\begin{theorem}[Alternative Imperative Equivalence]
    For all statements $s$ and states $\sigma, \sigma'$,
    \[
    \bstj {\langle s \mid \sigma \rangle} {\sigma'}
    \iff 
    \exists s'.\,\msj {\langle s \mid \sigma \rangle} {\langle s' \mid \sigma' \rangle}
    \]
\end{theorem}

\begin{proof}
This property follows from induction over the derivations.
\end{proof}

\end{document}
\endinput
